\documentclass{article}
\usepackage{graphicx} \usepackage{paper}
\usepackage{array}
\usepackage{multirow}
\usepackage{dsfont}

\usepackage{tikz}
\usepackage{amsmath}
\usepackage{caption}
\usepackage{bbm}
\usepackage{enumitem}

\newcommand{\bit}{\mathsf{b}}
\newcommand{\one}{\mathbbm{1}}

\newcommand{\PP}{\mathcal{P}}
\newcommand{\AAA}{\mathcal{A}}

\newcommand{\MMM}{\mathcal{M}}
\newcommand{\XXX}{\mathcal{X}}
\newcommand{\PPP}{\mathcal{P}}

\newcommand{\remove}[1]{}

\title{Is Randomness Necessary for Adaptive Data Analysis?%
\textsuperscript{\S}}

\author{
Edith Cohen\textsuperscript{*,$\dagger$} 
\and 
Haim Kaplan\textsuperscript{$\dagger$,*} 
\and 
Yishay Mansour\textsuperscript{$\dagger$,*} 
\and
Shay Sapir\textsuperscript{$\ddagger$,*} 
\and 
Uri Stemmer\textsuperscript{$\dagger$,*}
\and
Eliad Tsfadia\textsuperscript{$\P$}
}
\date{}

\theoremstyle{plain}

\newtheorem{remark}[definition]{Remark}  

\begin{document}
\maketitle
\begingroup
\renewcommand{\thefootnote}{\S}
\footnotetext{A preliminary version appeared on arxiv at July 2026, and it provided two limited answers: the bound $k=\tilde{O}(n)$ holds when 1)  the analyst and mechanism are given access to a random oracle, and 2) in a domain of tower-type size (using Ramsey's Theorem). This new version subsumes those limited answers using a domain of size $2^{\poly(n)}$. The new proof was obtained using ChatGPT (Astra 6).}
\endgroup

\begingroup
\renewcommand{\thefootnote}{} 
\footnotetext{
  \hspace{-7px}
  \textsuperscript{*}Google Research \quad
  \textsuperscript{$\dagger$}Tel Aviv University \quad
  \textsuperscript{$\ddagger$}Weizmann Institute of Science \quad
  \textsuperscript{$\P$}Bar-Ilan University
}
\endgroup

\begin{abstract}
The Adaptive Data Analysis (ADA) problem formalizes the challenge of preventing false discovery and overfitting when a dataset is repeatedly reused. Formally, our input is a dataset containing $n$ i.i.d.\ samples from an unknown distribution $\mathcal{P}$ over a domain $\mathcal{X}$, and our goal is to answer a sequence of $k$ adaptively chosen statistical queries with respect to $\mathcal{P}$. The main question is how many queries we can support (i.e., how large $k$ can be), primarily as a function of the number of samples $n$. This question has been intensively studied and is relatively well-understood for randomized mechanisms: there are computationally efficient mechanisms that support $k \approx n^2$ queries, and no computationally efficient mechanism can answer $k \gg n^2$ queries. In this paper, we address a fundamental question: is randomness necessary for ADA?

Despite a decade of work on ADA, this question remains open. A folklore observation dating back to the initial works on ADA is that randomness is {\em not} necessary when the analyst is computationally bounded. Yet, the necessity of randomness against computationally unbounded analysts has remained elusive. Our main contribution resolves this gap in the information-theoretic setting. Perhaps surprisingly, we show that randomness is {\em strictly necessary} to answer a non-trivial number of adaptive queries: when the analyst is unbounded, any deterministic mechanism can be forced to fail after just $k = \tilde{O}(n)$ queries.  

\end{abstract}

\section{Introduction}

Classical statistical theory for asserting the validity of a hypothesis tells us that the description of the hypothesis should be independent from the data on which it is evaluated. In practice, however, analysts frequently reuse the same dataset to answer multiple questions. Often, this exploration is adaptive: an analyst asks a statistical question, observes the answer, and uses that information to decide what to ask next. While this is how data analysis naturally happens, it makes it incredibly easy to overfit. Without careful intervention, the analyst will quickly find patterns that exist in the specific sample but fail to generalize to the true underlying distribution.

This problem was formalized by Dwork et al.~\cite{DworkFHPRR15} in what has come to be known as the {\em Adaptive Data Analysis (ADA)} problem. In this problem, a mechanism $\mathcal{M}$ holds a dataset $S$ consisting of $n$ samples drawn i.i.d.\ from an unknown distribution $\mathcal{P}$ over a domain $\mathcal{X}$. An analyst iteratively (and adaptively) submits a sequence of $k$ statistical queries $q : \mathcal{X} \to [0, 1]$. The mechanism $\mathcal{M}$ must answer these queries such that every answer is $\eps$-accurate (i.e., within $\pm\eps$) with respect to the true expectation of the query over $\mathcal{P}$.
To provide worst-case guarantees, the analyst is assumed to be adversarial, and is often referred to as an \emph{attacker}.  
The central question in ADA is establishing the optimal sample complexity: how large can $k$ be as a function of $n$?

This question is relatively well-understood for {\em randomized} mechanisms. A deep connection between ADA and various stability notions (such as differential privacy) has shown that randomized mechanisms can safely answer a large number of adaptive queries without overfitting \cite{DworkFHPRR15,DworkFHPRR15b,BassilyNSSSU16,RogersRST16,RussoZ16,FeldmanS17,FeldmanS18,FishRR18,LigettS19,SteinkeZ20,JungLN0SS20,ShenfeldL23,Blanc23}. Specifically, there are computationally efficient randomized mechanisms supporting $k \approx n^2$ queries, and, without making specific assumptions, no mechanism, even computationally unbounded, can answer $k \gg n^2$ queries.\footnote{If one assumes that the domain $\mathcal{X}$ is not too large, namely that $n\geq \polylog|\mathcal{X}|$, then there exist computationally {\em inefficient} mechanisms supporting $k \approx \exp(n/\polylog |\mathcal{X}|)$ queries. See Section~\ref{sec:related} (related works).}

However, all known constructions for the ADA problem rely fundamentally on the mechanism's ability to inject carefully calibrated random noise into its answers (again, unless one makes specific assumptions, under some of which deterministic mechanisms do exist). This reliance raises a natural question about the nature of adaptive data analysis: is randomness strictly necessary for ADA?

Resolving this question is fundamental to understanding the algorithmic toolkit required to prevent overfitting in the face of adaptive adversaries. This pursuit mirrors central lines of inquiry in other areas of theoretical computer science where separating the power of deterministic and randomized algorithms has been a long-standing theme. This includes works on communication complexity~\cite{yao1979some,kalyanasundaram1992probabilistic}, streaming algorithms \cite{flajolet1985probabilistic,alon1996space,harvey2008sketching,chakrabarti2016strong,kamath2021simple,kaplan2021separating,ChakrabartiS24}, dynamic algorithms \cite{BeimelKMNSS22,bateni2023optimal,BernsteinBFKS26}, and distributed computing \cite{ben1983another,fischer1985impossibility,lehmann1981advantages}.

\subsection{Our Contributions}

The primary contribution of this work is resolving the role of randomness for ADA.
We establish that to answer a super-linear number of adaptive queries, an algorithm must fundamentally rely on internal randomness. We demonstrate this by proving a strong separation: while deterministic mechanisms can match the performance of randomized ones against computationally bounded analysts, they are provably insufficient against unbounded analysts.

\paragraph{Baseline for bounded analysts.} To put our impossibility result in context, we first formalize a folklore observation (dating back
to the initial works on ADA) that randomness is not necessary when the analyst is computationally bounded. Specifically, by extracting a small number of random bits from the input sample and stretching them using cryptographic pseudorandom generators, we recover the same quadratic bound achievable by randomized mechanisms. Thus, when the analyst is bounded, randomness is a convenience rather than a necessity.

\begin{observation}[informal]\label{obs:intro}
Assuming that pseudorandom generators exist, there is a computationally efficient {\em deterministic} mechanism that, given $n$ i.i.d.\ samples from an unknown distribution $\PPP$, answers $k=\tilde{\Omega}(n^2)$ adaptively chosen statistical queries w.r.t.\ $\PPP$, provided that the entity selecting the queries is computationally efficient. 
\end{observation}

This observation was hinted at by some prior work without being stated or proven explicitly (see, e.g., \cite{stemmer2016}). Our contribution is in formalizing and proving it, which requires addressing some edge cases. The details are deferred to Appendix \ref{sec:bounded_adversary}.

\paragraph{Main result: Separation for unbounded analysts.} We now turn to the main technical focus of this paper, showing that when the analyst is computationally unbounded, the situation changes drastically. We prove that any deterministic mechanism can be forced to fail after just $k = \tilde{O}(n)$ queries. We establish this negative result in two stages:

\begin{itemize}
    \item {\bf Natural mechanisms:} We first consider a restricted class of deterministic mechanisms, called ``natural'' mechanisms, which observe only the evaluations of a query on the sample points. That is, a natural mechanism observes only $(q(x_1), \dots, q(x_n))$, where $x_1,\dots,x_n$ denote the sample points, rather than the full query description $q$. 
    Natural mechanisms were also considered by \cite{HardtU14,SteinkeU15,NissimSSSU18,NissimST23} as a steppingstone towards showing impossibility results for randomized mechanisms. These prior works showed that randomized natural mechanisms cannot answer more than $O(n^2)$ queries. In \Cref{sec:natural}, we present a completely different attack (and analysis) showing that {\em deterministic} natural mechanisms cannot answer more than $\tilde{O}(n)$ queries before failing.

    \item {\bf The plain model:} We then extend this barrier to general (i.e., not necessarily natural) deterministic mechanisms in \Cref{sec:plain-exponential}, using a domain of size $2^{\poly(n)}$, even when the deterministic mechanism is computationally unbounded.
\end{itemize}

Our main result is specified in the following theorem.

\begin{theorem}[informal]\label{thm:plain-exponential-intro}
For every sample size $n>2$, there is a domain $\X$ of size $2^{O(n^2\log n)}$ such that, for every deterministic mechanism $\MMM$ over $\X$, there exist a target distribution $\PP$ uniform on $20n$ points and a computationally unbounded analyst $\AAA$ that force $\MMM$ to incur error greater than $0.1$ within $O(n\log n)$ adaptive Boolean queries, with probability at least $9/10$.
\end{theorem}

\begin{remark}\label{rem:13}
An appealing aspect of the attacks presented by prior works was that they constructed a {\em single} attacker that defeats {\em every} mechanism (albeit requiring $n^2$ rounds). Our order of quantifiers is necessarily weaker: for every deterministic mechanism, we present a different attacker. This is unavoidable to beat the quadratic query bound. Indeed, if there were a single attacker that defeated every deterministic algorithm in $\tilde{O}(n)$ rounds, then by the law of total probability, this attacker would also defeat randomized algorithms, which is impossible as randomized algorithms can support $\tilde{\Omega}(n^2)$ queries.
\end{remark}

\paragraph{Extension to mechanisms with limited randomness.} Our techniques also extend to algorithms equipped with a limited amount of private randomness. Specifically, if a mechanism has access to $b$ random bits, we can conceptually view these bits as an extension of the empirical sample itself. By suitably adapting our attack, we show that it is possible to force such a mechanism to fail after $k = \tilde{O}(n + b)$ adaptive queries. This further clarifies the role of randomness in adaptive data analysis: to answer $\tilde{\Omega}(n^2)$ adaptive queries, an algorithm must fundamentally rely on a substantial amount of internal randomness. Any attempt to substitute this with deterministic operations or insufficient randomness ($b \ll n^2$) leaves the mechanism  vulnerable to rapid overfitting. We remark that $\tilde \Theta(n^2)$ random bits are known to be sufficient, see for example~\cite{BassilyNSSSU16,Blanc23}.

\begin{theorem}\label{thm:plain-limited-randomness}
For every sample size $n>2$ and integer $b\ge0$, there is a domain $\X$ of size $2^{O(n^2\log n+nb)}$ such that, for every mechanism $\MMM$ over $\X$ using at most $b$ private random bits, there exist a distribution $\PP$ uniform on $20n$ points and a computationally unbounded analyst that force an error greater than $0.1$ within $O(n\log n+b)$ adaptive Boolean queries, with probability at least $9/10$.
\end{theorem}

\subsection{Technical Overview for Theorem~\ref{thm:plain-exponential-intro}}

\subsubsection{A negative result for natural mechanisms}\label{sec:introNatural}
We first describe an attack against {\em natural} mechanisms (recall that a natural mechanism does not observe the full description of the queries it needs to answer; only their evaluations on the empirical sample points). Denote $N=\poly(n)$, let the target distribution be uniform over the domain $[N]$, and fix a {natural} deterministic mechanism $\MMM$. We construct a computationally unbounded attacker that throughout the execution maintains a set of ``surviving datasets'',  denoted $H_T$, containing all possible datasets of size $n$ that are consistent with the transcript $T$ of the interaction with $\MMM$ so far.
Initially, $H_T$ contains all possible datasets of size $n$ over $[N]$, so $|H_T|=N^n$. The attacker's goal is to rapidly shrink $H_T$ by issuing a sequence of adaptively chosen queries. We say that a query is {\em separating} if every valid answer (which must be close to the true distributional expectation, i.e., to the average over $[N]$) is forced to be inconsistent with a constant fraction of the datasets currently in $H_T$. Intuitively, as long as the mechanism reports valid answers, after $k = O(n \log N)$ such separating queries, the analyst will have learned the dataset, allowing it to present a query that highly overfits to the data, thereby forcing $\MMM$ to fail.

We now explain how the attacker finds such a separating query. For simplicity, let us focus on the first round of the interaction (so currently $H_T$ contains all datasets). The separating query our attacker chooses will be a {\em threshold query} over $N$, i.e., a query that evaluates to 0 on a suffix of the domain $[N]$. 
To this end, consider an $|H_T|\times (N+1)$ table, where cell $(i,j)$ contains the answer $\MMM$ would have returned had its dataset been $S_i$ and had it been presented with the threshold query $q_j$, defined as $q_j(x)=\mathds{1}_{\{x\leq j\}}$. See Figure~\ref{fig:naturalMatrix} for an illustration of how this table might look like.

\begin{figure}[htpb]
    \captionsetup{font={small, sf}, labelfont=bf}

    \hfill
    \begin{minipage}[c]{0.222\textwidth}
        \caption{Example of an $|H_T|{\times}(N{+}1)$ table. In this example, when the input dataset is $S_2$ and the query is $q_j$, the mechanism $\MMM$ returns the answer 0.74. The small fraction of coordinates that change between $q_{j-1}$ and $q_j$ are marked in bold.}
        \label{fig:naturalMatrix}

        \vspace{40px}
        
    \end{minipage}\hfill 
\begin{minipage}[c]{0.65\textwidth}
  
    \centering
    {\small
    \newcolumntype{C}[1]{>{\centering\arraybackslash}m{#1}}
    \renewcommand{\arraystretch}{1.25} 

\begin{tabular}{cc|c|c|c|c|c|c|}
        
\cline{3-8} 
        
\multirow{11}{*}{\rotatebox[origin=c]{90}{\textbf{Datasets in $\text{H}_{\text{T}}$}}} 
        & $S_{|H_T|}$  & $\approx 0$  & \dots\dots  & 0.71  & 0.71 & \dots  &  $\approx 1$  \\ \cline{3-8}
        & $S_{|H_T|\text{-}1}$  & $\approx 0$  & \dots\dots  & 0.77  & 0.77 & \dots  &  $\approx 1$  \\ \cline{3-8}
        & $S_{|H_T|\text{-}2}$  & $\approx 0$  & \dots\dots  & 0.75  & 0.75 & \dots  &  $\approx 1$  \\ \cline{3-8}
        & $S_{|H_T|\text{-}3}$  & $\approx 0$  & \dots\dots  & 0.69  & {\bf 0.76} &\dots  &  $\approx 1$  \\ \cline{3-8}
        & $S_{|H_T|\text{-}4}$  & $\approx 0$  & \dots\dots  & 0.68  & 0.68 & \dots  &  $\approx 1$  \\ \cline{3-8}
        & \vdots  & \vdots  & \dots\dots  & \vdots  & \vdots & \vdots  &  $\approx 1$  \\ \cline{3-8}
        & $S_5$  & $\approx 0$  & \dots\dots  & 0.72  & 0.72 & \dots  &  $\approx 1$  \\ \cline{3-8}
        & $S_4$  & $\approx 0$  & \dots\dots  & 0.81  & 0.81 & \dots &  $\approx 1$  \\ \cline{3-8}
        & $S_3$  & $\approx 0$  & \dots\dots  & 0.71  & {\bf 0.82} & \dots  &  $\approx 1$  \\ \cline{3-8}
        & $S_2$  & $\approx 0$  & \dots\dots  & 0.74  & 0.74 & \dots  &  $\approx 1$  \\ \cline{3-8}
        & $S_1$ & $\approx 0$ & \dots\dots & 0.68 & 0.68 & \dots &  $\approx 1$ \\ \cline{3-8}
        
\multicolumn{2}{c}{} & 
        \multicolumn{1}{c}{$q_0$} & \multicolumn{1}{c}{\dots\dots} & \multicolumn{1}{c}{$q_{j\text{-}1}$} & 
        \multicolumn{1}{c}{\phantom{-}$q_j$\phantom{1}} & \multicolumn{1}{c}{\dots} & 
        \multicolumn{1}{c}{$q_N$} \\[0.5ex] 

\multicolumn{2}{c}{} & \multicolumn{6}{c}{\textbf{Threshold queries}} \\ 
        
    \end{tabular}
    }
\end{minipage}
\end{figure}

Note that the attacker can compute this table before selecting its first query (by simulating $\MMM$). Also, observe that by the accuracy guarantees of $\MMM$, it must be that the first column is essentially $\vec{0}$ and the last column is essentially $\vec{1}$. The key observation regarding this table is that when going from one column $j-1$ to the next column $j$, only a small fraction of the coordinates might change. To see this, recall that the threshold queries $q_{j-1}$ and $q_j$ differ only in the $j$th point. Thus, as $\MMM$ is natural, it cannot notice this change unless its input dataset contains the point $j$. The fraction of datasets containing the $j$th point is at most $n/N$, which is small. So far we have established that:
\begin{enumerate}[itemsep=0px]
    \item[(1)] The left column is essentially all zeros.
    \item[(2)] The right column is essentially all ones.
    \item[(3)] Only a small fraction of the coordinates can change between every pair of consecutive columns.
\end{enumerate}
These three properties together imply the existence of a ``non-uniform'' column, meaning a column such that no value appears in, say, more than 90\% of its coordinates. To see this, let $j$ be the {\em first} column where more than 50\% of the coordinates are strictly bigger than 0.5. So in column $j-1$ at least $50\%$ of the coordinates are at most 0.5. As the columns change slowly, at least $(\frac{1}{2}-\frac{n}{N})$ of the coordinates in column $j$ must also be at most 0.5. So column $j$ is ``non-uniform'', and the corresponding query $q_j$ is {\em separating}: no matter what answer the attacker gets after querying $q_j$, this answer allows it to rule out $\approx 1/2$ of the datasets in $H_T$.

In the next round, in order to find the next separating query, the attacker constructs a new table where the y-axis contains only datasets that were not ruled out in previous rounds. We might also need to exclude a small number of points from the x-axis, corresponding to points that appear in ``too many'' of the surviving datasets. Otherwise, keeping such ``common'' points active might invalidate our key observation about the columns changing slowly. We show that there could be at most $O(n)$ such common points, and that excluding them has almost no effect on the calculation. Note that the values in the next table could be different than the values of the previous table, as the values returned by $\MMM$ could be history dependent.

Intuitively, by sequentially applying such separating queries, the attacker isolates the true dataset in at most $k = O(n \log N)$ rounds (unless the mechanism errs before that, in which case the attacker still wins).

\subsubsection{A na\"ive attempt at lifting this barrier to general deterministic mechanisms}\label{sec:naive}
Extending our impossibility result to general deterministic (not necessarily natural) mechanisms requires overcoming their ability to evaluate the full query description. The standard path for trying to achieve this would be for the attacker to encrypt the queries such that the mechanism could only decrypt the values corresponding to its empirical sample. More specifically, the target distribution would no longer be uniform over $[N]$. Instead, we first generate $N$ encryption keys $s_1,\dots,s_N$, and then define the target distribution to be uniform over the pairs $(j,s_j)$. Now, instead of directly issuing a query $q:[N]\rightarrow\{0,1\}$, the attacker encrypts $c_j={\rm Encrypt}(q(j),s_j)$ for every $j\in[N]$, and gives the vector of all encryptions $\vec{c}$ to the mechanism. The mechanism has the secret keys for points in its empirical sample, and can hence learn the value of $q$ restricted to its sample, but learns nothing on $q$ outside of its sample. Thus, the mechanism is effectively forced to behave like a natural mechanism. This strategy allowed the prior works of \cite{HardtU14,SteinkeU15,NissimSSSU18} to lift their impossibility result from (randomized) natural mechanisms to general mechanisms.

However, our attack against natural \emph{deterministic} mechanisms (defeating them in just $\tilde{O}(n)$ rounds) is fundamentally different from attacks presented in prior works. As a result, the strategy outlined above breaks down in our setting. Intuitively, the reason is that the mechanism could potentially use the ciphertexts and the (in-sample) keys as a source of randomness, undermining our assumption that the mechanism is deterministic. 
One might hope that the argument could still go through, because these keys and ciphertexts are known to the attacker, and are thus more akin to \emph{public} randomness rather than private randomness. Indeed, our attack would work seamlessly against a mechanism with public randomness, provided that it is natural for every realization of the public randomness. But this is not what the encryption-based strategy outlined above guarantees: changing an \emph{out-of-sample} query value, without changing the corresponding encryption key, alters the resulting ciphertext. This changes the description of the query, which can subsequently affect the mechanism's output. 
In other words, this encryption-based strategy only guarantees that the \emph{distribution} of answers given by the mechanism must be independent of out-of-sample query values \emph{over the randomness of generating the keys and encryptions}. It does not guarantee that the mechanism is natural for every fixed realization of the keys and ciphertexts, which is what is needed for the arguments from Section~\ref{sec:introNatural} to go through.

\subsection{Overcoming this challenge in the first round}\label{sec:overview_plain_first_round}
We overcome this challenge, and force any deterministic mechanism to behave as a ``bit-natural'' mechanism \emph{for the purpose of the attack}. 
A mechanism is called ``bit-natural'' if its discretized output $b=\one\{a\geq 0.5\}$ is a function of the query restricted only to the sample. That is, $b=b(q(x_1),\ldots,q(x_n))$, where $x_1,\ldots,x_n$ denote the sample points, and is independent of the query values on out-of-sample points. These mechanisms admit the same halving-style attack as natural mechanisms.
Our proof only forces bit-natural behavior on the threshold queries needed for the attack, and 
does not provide a general black-box reduction from the plain model to bit-natural mechanisms. Nevertheless, this suffices for the purpose of our proof.

The following probabilistic lemma is a key technical step. Its proof is given in \Cref{sec:proof_prob_bit_change}. The total variation distance between two discrete probability distributions $\mu_1,\mu_2$ over universe $\X$ is $TV(\mu_1,\mu_2)=\tfrac12\sum_{x\in\X}\big|\mu_1(x)-\mu_2(x)\big|$. We denote the distribution of a random variable $X$ conditioned on an event $A$ by $\cL(X\mid A)$.

\begin{lemma}\label{lem:prob_bit_change}
    Let $X$ be uniform over $\{-1,1\}^M$. Consider a function $f:\{-1,1\}^M\to\{-1,1\}^d$,  and let $Y=f(X)$. Then,
    \[
    \frac1M\sum_{u=1}^M TV\big(\cL(Y\mid X_u=-1),\cL(Y\mid X_u=1)\big)\leq \sqrt{2d/M}.
    \]
\end{lemma}

Intuitively, when $M \gg d$, \cref{lem:prob_bit_change} implies that $f(X)$ is almost independent of most of the $X_u$'s, and in particular, the average dependency on $X_1,\ldots,X_M$ is low.

We shall now describe the attack on the first round. First, construct the domain and distribution, similarly to the secret key picture.
Extend the domain to $[N]\times [M]$, and draw $N$ numbers $R_1,\ldots,R_N$  iid uniform over $[M]$, and set the distribution to be uniform over the pairs $(i,R_i)$. The analyst's queries $Q$ are now a full truth table over $[N]\times [M]$.
To issue a query $q:[N]\to\{0,1\}$, the analyst sets $Q(i,R_i)=q(i)$ for all $i\in[N]$, and completes the remaining coordinates with independent random bits.

Our goal is that when the analyst transitions from query $v_{i-1}\in\{0,1\}^N$ to query $v_i\in\{0,1\}^N$ that differ only at coordinate $i$, all datasets that omit element $(i,R_i)$ will have the same response.
Recall that the mechanism can inspect the entire truth table of the query. Thus, changing the query at one point $(i,R_i)$ may change its response on samples omitting that point. We will compensate by changing other truth-table entries as well. The aim is to sample consecutive threshold queries jointly so that their response bits agree on all samples omitting the changed support point, forcing the bit-natural behavior needed in this round. 
Towards this goal, we will use maximal coupling between the two response profiles. 
\begin{definition}[Maximal coupling]
    A maximal coupling between two distributions $\mu_1,\mu_2$ over space $\X$ is a distribution of a pair of random variables $Z_1,Z_2$ with the marginal distributions $Z_1\sim\mu_1$ and $Z_2\sim\mu_2$, that maximizes $\Pr[Z_1=Z_2]$.
\end{definition}
\begin{lemma}[Proposition 4.7 and Remark 4.8 of \cite{levinPeres2026markov}]\label{lem:max_coupling}
    A maximal coupling $(Z_1,Z_2)$ between two distributions $\mu_1,\mu_2$ exists and satisfies $\Pr[Z_1\neq Z_2]=TV(\mu_1,\mu_2)$.
\end{lemma}

Consider the responses (discretized to one bit) of datasets omitting $(i,R_i)$ when the input query is $Q$, and stack them in a single vector $V_i(Q)\in \{0,1\}^{d}$, where $d=(N-1)^n$. 
Denote by $Q_{i-1}$ and $Q_{i}$ the random completions of $v_{i-1}$ and $v_{i}$, respectively. 
Consider a maximal coupling between $V_i(Q_{i-1})$ and $ V_i(Q_{i})$ \emph{given $R_i$}, with the goal of bounding $\Pr[V_i(Q_{i-1})\neq V_i(Q_{i})]$.
It is instructive to consider a truth table $Q \in \{0,1\}^{N \times M}$ that is fixed on all rows $\neq i$, and is uniform over $\{0,1\}^M$ in the $i$-th row, and let $X\in\{-1,1\}^M$ represent that row (mapping $0\to-1$ and $1\to 1$).
Denote by $Y=f(X)$ the response profile. Now, $V_i(Q_{i-1})$ is distributed as $Y$ given $X_{R_i}=-1$ and $V_i(Q_{i})$ is distributed as $Y$ given $X_{R_i}=1$. By \Cref{lem:max_coupling},
\[
\Pr[V_i(Q_{i-1})\neq V_i(Q_{i})\mid R_i=u]\leq TV\big(\cL(Y\mid X_u=-1),\cL(Y\mid X_u=1)\big).
\]
Therefore,
\begin{align*}
    \Pr[V_i(Q_{i-1})\neq V_i(Q_{i})]&=\sum_{u=1}^M \Pr[R_i=u]\Pr[V_i(Q_{i-1})\neq V_i(Q_{i})\mid R_i=u] && \text{law of total probability} \\
    &\leq\frac1M \sum_{u=1}^M TV\big(\cL(Y\mid X_u=-1),\cL(Y\mid X_u=1)\big) && \text{$R_i$ is uniform over $[M]$} \\
    &\leq \sqrt{2d/M} && \text{by \Cref{lem:prob_bit_change}.}
\end{align*}
By a union bound, we obtain that with probability $1-N\sqrt{2N^n/M}$, the mechanism's responses to the threshold queries are bit-natural. Thus, the analyst can carry out the attack of \Cref{sec:introNatural} in the first round.

\begin{remark}
    The proof outlined here describes a coupling between subsequent threshold queries. Since the threshold queries form a path over the hypercube, it is straightforward to couple the entire sequence of threshold queries (which suffices for our proof). What does not follow, is a coupling of all $2^N$ possible queries (which would yield a black-box reduction from the plain model to natural mechanisms). To make such an argument, one need to ensure that if we complete a cycle in the hypercube, we end up with the same answers. If possible, this requires an additional argument which we deemed not necessary.
\end{remark}

\subsection{Generalizing to multiple rounds}

We briefly explain the obstacles for extending this proof to multiple rounds and how to overcome them, which enables the halving attack of \Cref{sec:introNatural}. Essentially, our argument for the first round used that the hidden locations $R=(R_1,\ldots,R_N)$ are independent of everything else, but this is no longer true.

\paragraph{Minor obstacle: common points may depend on the hidden locations $R$.}
Recall that in the attack on natural mechanisms, points that appear in many of the surviving datasets are excluded from the threshold queries. These, in turn, depend on the transcript and locations $R$. Thus, although we only have at most $N+1$ threshold queries in every round, they may depend on the hidden locations $R$.
This is a minor obstacle, which we overcome by augmenting the argument in \Cref{sec:overview_plain_first_round}: when analyzing two threshold queries that differ at coordinate $i$, we union bound over all $2^{N-1}$ possible completions of the remaining coordinates. This changes the failure probability to $N2^{N-1}\sqrt{2N^n/M}$. (Note that overall, we union bound over $N$ thresholds and $2^{N-1}$ completions of the remaining coordinates. This might be an overkill, but it has only a negligible effect on the universe size.)

\paragraph{Obstacle: the hidden locations $R$ may be revealed by the history.}
Consider some sequence of queries $h=(q_1,\ldots,q_{j-1})$ sent before round $j$.
These queries depend on the hidden locations $R$ and can reveal information about them. Therefore, given history $h$, the locations $R$ are no longer uniform, and thus the argument of the first round does not apply.

We overcome this obstacle by comparing to an hypothetical analyst that draws queries independently and uniformly at random from $\{0,1\}^{[N]\times[M]}$. We claim that the probability of any history of queries $h=(q_1,\ldots,q_{j-1})$ in the actual interaction, may be larger than a uniform draw by at most $D^{j-1}$ factor, where $D=(N+1)2^N$.
Clearly, if the past queries were uniform, no information about $R$ has leaked, and thus we could force the ``bit-natural'' behavior as in \Cref{sec:overview_plain_first_round}. Since the past queries are not uniform, the failure probability may increase, and we show in \Cref{subsec:exp-adaptive} it may increase by at most this $D^{j-1}$ factor.

The claim, bounding the probability of getting history $h=(q_1,\ldots,q_{j-1})$ in the actual interaction, follows since in each round of the actual interaction, the analyst chooses the values of the $N$ elements of the support (which has $2^N$ options), and chooses from a pool of $N+1$ queries. This may increase the probability of any query by at most $D=(N+1)2^N$ compared to a uniform draw. Accumulated over $j-1$ rounds, this is at most $D^{j-1}$.

\paragraph{Wrap-up.}
In the $j$-th round, the set of surviving datasets shrinks by a constant factor with probability at least $1-O(D^{j-1}2^{N-1}N\sqrt{N^n/M})$. Union bounding over $k=O(n\log N)=O(n\log n)$ rounds, the dominating term of the numerator is $D^k=2^{O(n^2\log n)}$. Therefore, for a sufficiently large $M=2^{O(n^2\log n)}$, we eliminate all datasets (and thus cause a failure) with constant probability.

\subsection{Other Related Works}\label{sec:related}

The literature on ADA is vast, with many exciting results beyond those discussed earlier. For example, researchers have explored variants of the ADA problem where the mechanism only needs to guarantee utility against ``weaker'' or ``more realistic'' adversaries, and shown impossibility results for these settings. Nissim et al.~\cite{NissimST23} proposed a ``balanced'' model that separates the entity choosing the target distribution from the attacker, thus removing the ability to embed secrets in the target distribution that only the attacker is privy to. They demonstrated that even in this more realistic setting, answering more than $\tilde{O}(n^2)$ queries is hard. In a different vein, Dinur et al.~\cite{DinurSWZ23} proved that existing lower bounds for ADA can be interpreted as arising from a {\em space} bottleneck rather than a sampling bottleneck. 

A growing body of literature considers the challenges of ADA with {\em correlated} input samples. Kontorovich et al.~\cite{KontorovichSS22} extended stability-based generalization guarantees to non-i.i.d.\ settings by formally quantifying the degree of correlation between samples. Rapoport et al.~\cite{RapoportCS25} recently established tight lower bounds for these settings, proving an unavoidable utility gap: under natural conditions, one can answer at most $O(n)$ adaptively chosen queries on correlated data (even using randomized mechanisms), in sharp contrast to the $\tilde{O}(n^2)$ limit in the independent setting.

In our work, we focus on the vanilla, unrestricted formulation of the ADA problem. As we mentioned, if we assume that the domain $\XXX$ is not too large, namely that $n\geq\polylog|\XXX|$, then answering more than $n^2$ queries becomes possible (albeit inefficiently), even using {\em deterministic} mechanisms. Specifically, Dwork et al.~\cite{DworkFHPRR15} and Bassily et al.~\cite{BassilyNSSSU16} showed that there are (inefficient) mechanisms that answer $\exp(n/\sqrt{\log|\XXX|})$ adaptive queries using a sample of size $n$. Dwork et al.~\cite{DworkFHPRR15b} showed that a qualitatively similar positive result can be obtained using an (inefficient) {\em deterministic} mechanism. Recently, this query bound of $\exp(n/\sqrt{\log|\XXX|})$ was shown to be tight by Lyu and Talwar~\cite{LyuT25}.

At a high level, beneath the machinery required to instantiate it in our setting, our attack is driven by a search-by-halving strategy. Isolating a hidden object by repeatedly shrinking a space of possibilities is a recurring theme across theoretical computer science, appearing in the Halving algorithm for online learning~\cite{littlestone1988learning}, query-based learning~\cite{angluin1988queries}, randomness lower bounds for adversarially robust streaming~\cite{ChakrabartiS24,stoeckl2023streaming}, oracle attacks in cryptography~\cite{bleichenbacher1998chosen,manger2001chosen}, and binary search with erroneous answers~\cite{rivest1980coping,pelc2002searching}.

\section{Preliminaries}

We first define the Adaptive Data Analysis (ADA) framework, distributional accuracy, and natural algorithms.

\begin{definition}[Adaptive Data Analysis (ADA) Framework]
In the ADA framework, a mechanism $\MMM$ interacts with a data analyst $\AAA$ over $k$ rounds. Given a dataset $S = (x_1, \dots, x_n) \sim \mathcal{P}^n$ drawn i.i.d. from a distribution $\mathcal{P}$ over a domain $\mathcal{X}$, in each round $j \in [k]$:
\begin{enumerate}
    \item The analyst $\AAA$ outputs a statistical query $q_j: \mathcal{X} \to [0, 1]$, which may depend on previous queries and answers.
    \item The mechanism $\MMM$ observes $q_j$ and outputs an answer $a_j \in \mathbb{R}$.
\end{enumerate}
\end{definition}

\begin{definition}[$\eps$-Accuracy]\label{def:epsAcc}
Let $\eps>0,k,n,N,\X$.
A mechanism $\MMM$ is $\eps$-accurate for $k$ adaptive queries given a sample of size $n$, if for each distribution $\mathcal{P}$ over $\X$ and analyst $\AAA$,
with probability at least $9/10$ over the sample $S \sim \mathcal{P}^n$, and the randomness of $\MMM$, the answers $a_1, \dots, a_k$ satisfy
\[ \max_{j \in [k]} \left|a_j - \mathbb{E}_{x \sim \mathcal{P}}[q_j(x)]\right| \le \eps. \]
\end{definition}

\begin{definition}[Natural Mechanism]
A mechanism $\MMM$ is called \emph{natural} if, when given a dataset $S = (x_1, \dots, x_n)$, then for every given statistical query $q : \mathcal{X} \to [0, 1]$, the output of the mechanism $\MMM$ is only a function of the query evaluations on the sample points, i.e., $(q(x_1), \dots, q(x_n))$. (For randomized mechanisms, the output is also a function of the internal coins.)
\end{definition}

\section{Natural Mechanisms}\label{sec:natural}
We first provide an attack against natural deterministic mechanisms.
Suppose our domain is $\X=[N]$, where $N=20n$. (If $|\X|$ is larger, just restrict the distribution to be supported only on the first $N$ elements of $\X$.)
Consider a natural deterministic mechanism $\MMM$ that, given a sample of size $n$, is $0.1$-accurate for $k$ adaptive queries by a computationally unbounded analyst. We aim to show that $k=O(n\log N)$. Consider the uniform distribution over $[N]$.

An unbounded adversary $\AAA$ can simulate the answers of all possible datasets. Thus, for a transcript $T=(q_1,a_1,\dots,q_i,a_i)$, the adversary can compute the set of datasets that are consistent with $T$, denoted $H_T\subseteq[N]^n$. Furthermore, $\AAA$ can compute the answers of datasets in $H_T$ for any query $q$. The adversary aims to find a query for which every answer $a$ is consistent with at most $\frac{9}{10}$ fraction of the datasets in $H_T$, called a {\em separating query}. The next lemma shows that $\AAA$ can indeed find a separating query, and moreover, that query is binary $q:[N]\to \{0,1\}$.

\begin{lemma}\label{lem:exists_separating_query}
Let a deterministic natural mechanism $\MMM$ and a transcript $T=(q_1,a_1,\dots,q_i,a_i)$. There exists a query $q:[N]\to\{0,1\}$ that either causes the mechanism to have error $>0.1$, or is a separating query for $\MMM$.
\end{lemma}

The proof uses \emph{threshold queries}, defined as follows.

\begin{definition}[Threshold query]\label{def:threshold_query}
Let $X_m = \{x_1, \dots, x_m\} \subset [N]$ be a set of  points. Let the remaining elements of the domain be arbitrarily ordered as $z_1, z_2, \dots, z_{N-m}$. A \emph{threshold query} is any query of the form,
\begin{equation*}
    q_j^{X_m}(x) = 
    \begin{cases} 
      1 & \text{if } x \in X_m \cup \{z_1, \dots, z_j\} \\
      0 & \text{otherwise.}
    \end{cases}
\end{equation*}
\end{definition}

The set $X_m$ as defined in \Cref{def:threshold_query} is arbitrary, but in the proof of \Cref{lem:exists_separating_query}, it will be  a set of $m=O(n)$ points that appear in most of the databases in $H_T$.

\begin{proof}[Proof of \Cref{lem:exists_separating_query}]
    Let $X_m=\{x_1, \dots, x_m\} \subset [N]$ be the set of all points $x\in [N]$ s.t. $|\{S\in H_T:x\in S\}|\ge\frac{8}{10}|H_T|$, i.e., $x$ is in most of the surviving datasets. We claim that $m=O(n)$, by the following counting argument.

    We  calculate the number of points in $H_T$, including multiplicities, in two ways. First, every $S\in H_T$ contains $n$ points, so we have precisely $n|H_T|$ points. On the other hand, each of $x_1,\dots,x_m$ is in $\frac{8}{10}|H_T|$ of the datasets in $H_T$, thus contributes at least that number of copies to the overall number. Summing over these $m$ common points, we obtain 
    \begin{equation}\label{eq:common_points}
        \frac{8}{10}|H_T|\cdot m\leq|H_T|\cdot n,
    \end{equation}
    and thus $m\leq \tfrac{5}{4}n$.

    Consider the query that assigns $1$ to $x_1,\dots,x_m$, and $0$ to all other points, which is the threshold query $q_0^{X_m}$. The answer $a$ to query $q_0^{X_m}$ must satisfy $a\le0.1+\frac{m}{N}\le0.2$. If $\leq \tfrac{9}{10}$ fraction of the datasets in $H_T$ return answer at most $0.2$, then the query $q_0^{X_m}$ satisfies the conclusion of the theorem, so suppose more than $9/10$ fraction of the datasets in $H_T$ return a valid answer to $q_0^{X_m}$.
    Similarly, the answer to the query that assigns $1$ to all points (which is $q_{N-m}^{X_m}$) is at least $0.9$. If $\leq \tfrac{9}{10}$ fraction of the datasets in $H_T$ return answer at least $0.9$, then the query $q_{N-m}^{X_m}$ satisfies the conclusion of the theorem, so suppose more than $9/10$ fraction of the datasets in $H_T$ return a valid answer to $q_{N-m}^{X_m}$.
    Therefore, there  must be some $j\in[0,N-m]$, for which 
    most datasets in $H_T$ answer $\le0.5$ for $q_j^{X_m}$, but most datasets in $H_T$ answer $>0.5$ for $q_{j+1}^{X_m}$.
    
    The analyst simulates the answer of every $S\in H_T$ for every threshold query to locate $q_j^{X_m}$.
    By construction, the point $z_{j+1}$ is not common. Since the mechanism is natural, only datasets that contain $z_{j+1}$ may have different answers for the queries $q_j^{X_m}$ and $q_{j+1}^{X_m}$. 
    
    We claim that either $q_j^{X_m}$ or $q_{j+1}^{X_m}$ is a separating query, which concludes the proof.
    Suppose that $q_j^{X_m}$ is not a separating query, meaning that at least $\tfrac{9}{10}$ fraction of the datasets in $H_T$ return some answer $a \le 0.5$ for the query $q_j^{X_m}$. At most $\tfrac{8}{10}|H_T|$ of these contain $z_{j+1}$ (since $z_{j+1}$ is not common), and thus the most common answer for $q_{j+1}^{X_m}$, which is $>0.5$, is returned by at most $\tfrac{9}{10}|H_T|$ datasets, so $q_{j+1}^{X_m}$ is a separating query.
\end{proof}

\begin{theorem}
A natural deterministic mechanism $\MMM$ that is $0.1$-accurate can answer at most $k = O(n \log N)$ queries from a computationally unbounded analyst.
\end{theorem}
\begin{proof}
By \Cref{lem:exists_separating_query}, as long as the mechanism returns correct answers, in every round, the analyst can find a separating query. Initially, the number of possible datasets is $ N^n$. Every such query rules out at least a constant factor of the datasets. Therefore, the analyst can fail the mechanism after $O(n\log N)$ queries. (Observe that if $|H_T|=1$, then $X_m=S$, hence the mechanism must reply the same for $q_0^{X_m}$ and $q_{N-m}^{X_m}$, and thus errs on at least one of these queries.) 
\end{proof} 

\begin{remark}
    The proof only compares the mechanism's answers to $0.5$, and thus extends to a larger family of mechanisms.
    We say that a mechanism is ``bit-natural'' if by discretizing its outputs by $\one\{a\geq 0.5\}$, we obtain a natural mechanism.
    We will use the following in \Cref{sec:plain-exponential}.
\end{remark}
\begin{lemma}[``Bit-natural'' mechanisms]\label{lem:bit_natural}
    A ``bit-natural'' mechanism that is $0.1$-accurate can answer at most $O(n\log N)$ queries from a computationally unbounded analyst.
\end{lemma}

\section{A Plain-Model Lower Bound}\label{sec:plain-exponential}
We now extend the lower bound on deterministic natural mechanisms to deterministic mechanisms in the plain model, i.e., where the mechanism's answer may depend on the entire truth table of the query, proving \Cref{thm:plain-exponential-intro}.
The proof is not by a black-box reduction --- it forces ``bit-natural'' behavior on threshold queries, which suffices to carry out the attack of \Cref{lem:bit_natural}.

\subsection{Key probabilistic lemma}\label{sec:proof_prob_bit_change}
Before describing the attack, we first prove our key technical lemma (\Cref{lem:prob_bit_change}). Recall \Cref{lem:prob_bit_change}: $X$ is uniform over $\{-1,1\}^M$, we have a function $f:\{-1,1\}^M\to\{-1,1\}^d$ and $Y=f(X)$.
The lemma states that 
\[
\frac1M\sum_{u=1}^M \operatorname{TV}\bigl(
\mathcal L(Y\mid X_u=-1),
\mathcal L(Y\mid X_u=1)
\bigr)
\le \sqrt{\frac{2d\ln 2}{M}}.
\]

\begin{proof}[Proof of \Cref{lem:prob_bit_change}]
For every $u\in[M]$ and $y\in\{-1,1\}^d$ with $\Pr[Y=y]>0$,
define $c_u(y)\in\{-1,1\}$ to be the more likely
value of $X_u$ conditioned on $Y=y$, breaking ties
arbitrarily. That is,
\[
c_u(y)=\operatorname{sign}
\bigl(\Pr[X_u=1\mid Y=y]
      -\Pr[X_u=-1\mid Y=y]\bigr),
\]
where $\operatorname{sign}(0)=1$.
These form a string
$c(y)=(c_1(y),\ldots,c_M(y))\in\{-1,1\}^M$. Let $\C$ be the set of such strings. Since $y\in\{-1,1\}^d$, we have $|\C|\leq 2^d$. Thus,

\begin{align*}
    &\operatorname{TV}\bigl(
\mathcal L(Y\mid X_u=-1),
\mathcal L(Y\mid X_u=1)
\bigr)\\
&= \frac12\sum_{y} \big| \Pr[Y=y\mid X_u=1]-\Pr[Y=y\mid X_u=-1]\big| \\
&= \sum_{y} \big| \Pr[Y=y, X_u=1]-\Pr[Y=y, X_u=-1]\big| && \text{$X_u$ is a fair sign}\\
&=\sum_{y} \big| \Pr[X_u=1\mid Y=y]-\Pr[X_u=-1\mid Y=y]\big|\Pr[Y=y]\\
&=\sum_{y} c_u(y)\E[X_u \mid Y=y]\Pr[Y=y] && \text{by definition of $c_u(\cdot)$}\\
&=\E[X_u c_u(Y)] && \text{law of total probability.}
\end{align*}

To bound the sum over $u$, we use Massart's Lemma~\cite[Theorem~3.7]{MRT18}: for a finite nonempty set $A\subseteq\mathbb R^M$ and a vector $X$ of independent fair signs,
\[
 \E\max_{a\in A}\langle X,a\rangle
 \le\left(\max_{a\in A}\|a\|_2\right)\sqrt{2\ln|A|}.
\]
Every $c\in\mathcal C$ has norm $\sqrt M$. Since the realized guess belongs to $\mathcal C$, the preceding inequality and Massart's lemma give
\begin{align*}
 \sum_{u=1}^M \operatorname{TV}\bigl(
\mathcal L(Y\mid X_u=-1),
\mathcal L(Y\mid X_u=1)
\bigr)
 &=\E\langle X,c(Y)\rangle \\
 &\le\E\max_{c\in\mathcal C}\langle X,c\rangle
 \le\sqrt{2M\ln|\mathcal C|}
 \le\sqrt{2Md\ln2}.
\end{align*}
Therefore,
\begin{align*}
 \frac1M \sum_{u=1}^M \operatorname{TV}\bigl(
\mathcal L(Y\mid X_u=-1),
\mathcal L(Y\mid X_u=1)
\bigr) 
\leq  \sqrt{\frac{2d\ln2}{M}}\leq \sqrt{\frac{2d}{M}},
\end{align*}
concluding the proof of \Cref{lem:prob_bit_change}.
\end{proof}

\subsection{The distribution}

Fix a sample size $n>2$ and set
\[
 N=20n,\qquad k=\left\lceil\log_{10/9}N^n\right\rceil+1=O(n\log n).
\]
Consider the domain $\X=[N]\times[M]$, where $M$ will be chosen at the end of the section. 
%
%
For every $i\in[N]$, independently choose one uniformly random location $R_i\in[M]$ for the support point $(i,R_i)$, and write $R=(R_1,\ldots,R_N)$ for the vector of locations. Set the distribution $\PP$ to be uniform on the $N$ points $(1,R_1),\ldots,(N,R_N)$. The analyst knows the locations $R$.

\paragraph{Queries}
Fix an arbitrary deterministic mechanism $\MMM$ on this domain. 
We restrict the queries to threshold queries over the $N$ support points, and extend them to the whole universe by filling all other entries with independent fair bits. Denote by $\mu$ the uniform distribution on $\{0,1\}^{N\times M}$. 
For every realization $r=(r_1,\ldots,r_N)$ of $R$ and
query $v=(v_1,\ldots,v_N)\in\{0,1\}^N$ over the support points $(1,r_1),\ldots,(N,r_N)$, denote by $\mu_{r,v}$ the distribution of the resulting $N\times M$ truth table. That is, suppose $Q\sim\mu$. Conditioned on $Q(i,r_i)=v_i$ for all $i\in[N]$, we have $Q\sim\mu_{r,v}$.
Throughout, we use $Q$ for a random query and $q$ for a realized query.

\subsection{Forcing ``bit-natural'' responses in a single round}\label{subsec:exp-coupling}
We first analyze a single round.
Fix a round $j$ and a list $h=(q_1,\ldots,q_{j-1})$ of past query tables. We analyze this round with $h$ held fixed and the support locations sampled uniformly. This is an auxiliary experiment; conditioned on the actual interaction reaching $h$, the locations are not necessarily uniform. We account for that dependence in \Cref{subsec:exp-adaptive}.

Recall that in the plain model, the mechanism can inspect the entire truth table of the query. Thus, changing the query at one point may change its response on samples omitting that point. We will compensate by changing other truth-table entries as well. The aim is to sample consecutive threshold queries jointly so that their response bits agree on all samples omitting the changed support point, forcing the bit-natural behavior needed in this round.

\paragraph{Recording all responses together.}
Fix a row $i$. When the threshold turns on $(i,R_i)$, we need the response bits to agree \emph{simultaneously} for every surviving sample omitting $i$. We therefore record all these responses in one vector. Moreover, we include \emph{all} samples omitting $i$, even those already ruled out. This makes the set of vector coordinates independent of the transcript and of the hidden location $R_i$.

For a candidate sample with row indices $s\in[N]^n$, round the mechanism's answer $a$ when given query $q$ to the bit $\bit_s(q;h)=\one\{a\geq0.5\}$.
The response profile is the vector with one bit for each sample of size $n$ using only rows other than $i$:
\begin{equation}\label{eq:exp-profile}
 V_i^h(q)=\bigl(\bit_s(q;h)\bigr)_{s\in([N]\setminus\{i\})^n}
 \in\{0,1\}^{(N-1)^n}.
\end{equation}
For fixed $R_{-i}=r_{-i}$ and $h$, this vector-valued function of $q$ does not depend on $R_i$: none of its candidate samples contains $(i,R_i)$. Here $R_{-i}$ denotes all locations except $R_i$. The history $h$ contains only the past queries, not the past answers, which may reveal $R_i$ (and determine which candidates survive).
We now use \Cref{lem:prob_bit_change} to prove the following,
which shows that we can jointly sample the threshold queries and guarantee that $V_i^h$ is fixed (basically forcing bit-natural behavior in this round).

\begin{lemma}[Coupling a threshold sequence]\label{lem:exp-coupling}
Fix the past query tables $h$ and draw $R\sim\operatorname{Unif}([M]^N)$. Let $v^0,\ldots,v^N\in\{0,1\}^N$ be threshold patterns from \Cref{def:threshold_query}, allowed to depend arbitrarily on $R$. Write $i_\ell$ for the coordinate changed between $v^{\ell-1}$ and $v^\ell$.
There is a rule for jointly sampling full query tables $Q_0,\ldots,Q_N$ satisfying the following.
For every $r$ and $\ell\in\{0,\ldots,N\}$, conditional on $R=r$, we have $Q_\ell\sim \mu_{r,v^\ell}$, and
\[
 \Pr_{R,Q_0,\ldots,Q_N}\bigl[\forall \ell\in[N], \quad V_{i_\ell}^h(Q_{\ell-1})=V_{i_\ell}^h(Q_\ell)
          \bigr]
 \ge 1-N2^{N-1}\sqrt{\frac{2N^n}{M}}.
\]
\end{lemma}

\begin{proof}
We first describe the rule for jointly sampling $Q_0,\ldots,Q_N$. Fix $R=r$, and for every $\ell\in\{0,\ldots,N\}$, let $Q_\ell\sim \mu_{r,v^\ell}$.
For every $\ell\in[N]$, we take a maximal coupling of $V_{i_\ell}^h(Q_{\ell-1})$ and $V_{i_\ell}^h(Q_{\ell})$. A maximal coupling of two random variables is a joint draw whose disagreement probability equals the total-variation distance between the distributions. 
To actually perform the sampling, first draw $Q_0\sim\mu_{r,v^0}$. At step $\ell$ given $V_{i_\ell}^h(Q_{\ell-1})$, sample the profile $V_{i_\ell}^h(Q_\ell)$ according to the coupling, and then sample $Q_\ell$ given $V_{i_\ell}^h(Q_\ell)$.

\paragraph{Bounding failure for a fixed pair.}
The changed row $i_\ell$ and unchanged pattern $v_{-i_\ell}$ at a step may depend on $R$. We therefore consider each fixed row $i\in[N]$ and pattern $w=(w_t)_{t\ne i}\in\{0,1\}^{N-1}$ separately. Define the event
\[
 F_{i,w}=\bigl\{\exists\ell\in[N]:\quad
 i_\ell=i,\quad v^\ell_{-i}=w,\quad
 V_i^h(Q_{\ell-1})\ne V_i^h(Q_\ell)\bigr\}.
\]
The event is false if the sequence never uses $(i,w)$; otherwise, because each row changes at most once, it concerns a single step.
Fix $i,w,r_{-i}$.
We bound $\Pr[F_{i,w}\mid R_{-i}=r_{-i}]$ by averaging over the uniform location $R_i$. For $u\in[M]$, consider
\[
 \Pr[F_{i,w}\mid R_{-i}=r_{-i},R_i=u].
\]
To analyze these probabilities, consider a random table $Q$ with $Q(t,r_t)=w_t$ for every $t\ne i$ and all other entries independent fair bits. Its distribution does not depend on $u$: the entire row $i$ remains unconditioned. Write $\xi_u=2Q(i,u)-1$, so $\xi_1,\ldots,\xi_M$ are independent fair signs. For fixed $R_{-i}=r_{-i}$ and $h$, the map $q\mapsto V_i^h(q)$ also does not depend on $u$.

Fix $R_i=u$, so all locations and the threshold sequence are now fixed. If the sequence uses $(i,w)$, let $\ell$ be its unique corresponding step. The two support patterns agree with $w$ outside $i$, and satisfy $v_i^{\ell-1}=0$ and $v_i^\ell=1$. Thus the reference table $Q$ satisfies
\[
 \mathcal L(Q\mid\xi_u=-1)=\mu_{r,v^{\ell-1}},
 \qquad
 \mathcal L(Q\mid\xi_u=1)=\mu_{r,v^\ell},
\]
where $\mathcal L(X)$ denotes the law of the random variable $X$.
By construction, the actual profiles $V_i^h(Q_{\ell-1})$ and $V_i^h(Q_\ell)$ have a maximal coupling of these two laws. 
To upper bound their disagreement probability, consider the following alternative
coupling. First, sample and fix all entries of $Q$
in rows other than $i$, so that only the $i$-th row
remains random. Then, conditioned on these entries,
use a maximal coupling of the two response profiles.
Since the original maximal coupling minimizes the
probability of disagreement among all couplings,
its disagreement probability is at most that of
this alternative coupling.
%
Now, $V_i^h(Q)$ is only a function of the $i$-th row of $Q$, or alternatively, of $\xi$. Denote $Y=Y(\xi)=V_i^h(Q)$.
By the maximum coupling lemma (\Cref{lem:max_coupling}),
\begin{align*}
 \Pr[F_{i,w}\mid R_{-i}=r_{-i},R_i=u]
 &=\Pr\bigl[V_i^h(Q_{\ell-1})\ne V_i^h(Q_\ell)\mid R_{-i}=r_{-i},R_i=u\bigr]\\
  &\leq TV\big(\cL(Y\mid\xi_u=1),\cL(Y\mid\xi_u=-1)\big).
\end{align*}
If the threshold sequence does not use $(i,w)$, the left-hand side is zero, and this inequality still holds.

We can now apply \Cref{lem:prob_bit_change}. Recall $\xi$ is uniform over $\{-1,1\}^M$, and $Y\in\{0,1\}^{(N-1)^n}$ is only a function of $\xi$.
Therefore,
\[
\frac{1}{M}\sum_{u=1}^M TV\big(\cL(Y\mid\xi_u=1),\cL(Y\mid\xi_u=-1)\big)\leq \sqrt{\frac{2N^n}{M}}.
\]

As $R_i$ is uniform, this implies
\[
 \Pr[F_{i,w}\mid R_{-i}=r_{-i}]
 =\frac1M\sum_{u=1}^M \Pr[F_{i,w}\mid R_{-i}=r_{-i},R_i=u]
 \le\sqrt{\frac{2N^n}{M}}.
\]
The bound holds for every $r_{-i}$, so averaging over the other locations gives
\begin{equation}\label{eq:exp-coupling}
 \Pr[F_{i,w}]\le\sqrt{\frac{2N^n}{M}}.
\end{equation}

\noindent\textbf{Completing the union bound.}
Let $E$ be the event that some adjacent pair disagrees. 
Each failing step has a changed row $i$ and unchanged pattern $w$, so $E=\bigcup_{i,w}F_{i,w}$. A union bound and Eq. \ref{eq:exp-coupling} give

\begin{equation}\label{eq:exp-average-failure}
 \Pr[E]
 \le\sum_{i=1}^N\sum_{w\in\{0,1\}^{N-1}}\Pr[F_{i,w}]
 \le N2^{N-1}\sqrt{\frac{2N^n}{M}}.
\end{equation}
This completes the proof of \Cref{lem:exp-coupling}.
\end{proof}

\subsection{Accounting for adaptivity and completing the proof}\label{subsec:exp-adaptive}

We now run the threshold-query attack underlying \Cref{lem:bit_natural} on the support indices. In each round, the transcript determines $H_T$, its common set $X_m$, and the threshold sequence (ordering the remaining indices increasingly). Complete this sequence using \Cref{lem:exp-coupling}, and select the first completed query for which every $0.1$-accurate answer leaves at most $9|H_T|/10$ surviving samples. If there is no such query, send $Q_0$ as a fallback. The proof of \Cref{lem:bit_natural} uses only agreement of adjacent response bits on samples omitting the changed point, so the fallback is needed only when a coupling fails. Denote this analyst by $\AAA$. It uses fresh coupling coins in each round and retains only $R$ and the actual transcript, discarding the unused coupling data.


In the $j$-th round,
for a fixed query history $h=(q_1,\ldots,q_{j-1})$, \Cref{lem:exp-coupling} bounds the probability of a coupling failure over uniform support locations $R$. The actual past queries $h$ depend on the locations $R$ and can reveal information about them. We first bound the probability of each history by $D^{j-1}$ times its uniform probability, where $D=(N+1)2^N$. We then compare the actual interaction with an analyst that submits that history unconditionally and apply the single-round bound.
Note we fixed the past query tables and not the full transcript including answers. The former suffices, since for each (surviving) dataset, its answers can be reconstructed deterministically.

\subsubsection{Comparing queries and histories with uniform tables}

To control adaptivity, we compare the distribution of the query actually sent with $\mu$, the uniform distribution on $N\times M$ Boolean tables. The next lemma bounds the probability of every table event by $D$ times its uniform probability, regardless of how the analyst selects among the completed queries. The bound uses only their marginals, so it also holds when a profile coupling fails.

\begin{lemma}[Selected-query density]\label{lem:exp-density}
Fix the locations $R=r$ and at most $N+1$ support patterns $v^0,\ldots,v^L$. Let $Q_0,\ldots,Q_L$ have any joint distribution with $Q_\ell\sim\mu_{r,v^\ell}$ for each $\ell$. For any selection $J\in\{0,\ldots,L\}$ and every set $A$ of tables,
\begin{equation}\label{eq:exp-density}
 \Pr[Q_J\in A]\le D\mu(A),
 \qquad D=(N+1)2^N.
\end{equation}
The selection $J$ may depend on the whole sampled sequence. The statement also holds conditional on any fixed past transcript for which the stated marginals hold.
\end{lemma}
\begin{proof}
Recall that $\mu_{r,v}$ is the law of a uniform table conditioned on $Q(i,r_i)=v_i$ for every $i\in[N]$. These $N$ prescribed bits have probability $2^{-N}$ under $\mu$, so $\mu_{r,v}(A)\le2^N\mu(A)$. Regardless of how $J$ depends on the sequence,
\[
 \Pr[Q_J\in A]
 \le\sum_{\ell=0}^L\Pr[Q_\ell\in A]
 =\sum_{\ell=0}^L\mu_{r,v^\ell}(A)
 \le(N+1)2^N\mu(A).
\]
\end{proof}

\paragraph{From one query to the full history.}

Write $q^{<j}=(q_1,\ldots,q_{j-1})$ for the full query tables sent before round $j$, and let $h$ denote a fixed possible value of this history. The quantity $\mu^{\otimes(j-1)}(h)$ is the probability of obtaining $h$ by drawing these tables independently and uniformly. Denote by $s^*\in [N]^n$ the \emph{index sample}. That is, $s^*\sim\operatorname{Unif}([N]^n)$ independently of the uniform locations $R$, and the actual sample is $S=\big((s_\ell^*,R_{s_\ell^*})\big)_{\ell=1}^n$.

Fix $R=r$ and the actual index sample $s^*$ that was given as input to the mechanism. Given the past full tables, the deterministic mechanism's earlier answers are determined. Thus the transcript determines the set of surviving datasets as in \Cref{sec:natural}, and hence the threshold patterns $v^0,\ldots,v^{N}$. The fresh completed tables have the corresponding marginals $\mu_{r,v^\ell}$, so \Cref{lem:exp-density} bounds the law of the selected table. Multiplying this bound over rounds gives, for every history $h$,
\begin{equation}\label{eq:exp-history}
 \Pr[q^{<j}=h\mid R=r,s^*]\le D^{j-1}\mu^{\otimes(j-1)}(h).
\end{equation}

\subsubsection{Bounding the probability of a coupling failure}\label{subsec:exp-failure}

Let $E_j$ be the event that at least one adjacent profile equality fails in round $j$ of the interaction with $\AAA$. We compare this interaction with one in which the past queries are prescribed in advance. This preserves the uniform distribution of the support locations and lets us apply \Cref{lem:exp-coupling} directly.

\paragraph{An analyst with a fixed history.}
Fix $h=(q_1,\ldots,q_{j-1})$. Define $\AAA_h$ to submit these $j-1$ queries regardless of the locations and the mechanism's answers, while recording those answers. In round $j$, it follows $\AAA$'s rule on the resulting transcript, using fresh coupling randomness. The recorded answers determine $H_T$ and its threshold sequence. This rule is well-defined even when $h$ is unreachable under $\AAA$: the actual sample is consistent with the transcript, so $H_T$ is nonempty.

Let $F_h$ be the event that an adjacent profile equality fails in round $j$ of this experiment. For every fixed index sample $s^*=s$, the locations remain uniform, and the deterministic mechanism's answers to $h$ specify a threshold sequence as a function of $r$. Thus \Cref{lem:exp-coupling} bounds $\Pr[F_h\mid s^*=s]$. Averaging over the independent index sample gives
\[
 \Pr[F_h]\le N2^{N-1}\sqrt{\frac{2N^n}{M}}.
\]

\paragraph{Comparing the two interactions.}
Run the interactions with $\AAA$ and $\AAA_h$ on the same locations and sample, and use the same fresh coupling coins in round $j$, independently of the coins used in earlier rounds. The notation $q^{<j}$ continues to denote the history of $\AAA$. On the event $q^{<j}=h$, the two executions have the same transcript and mechanism state, hence the same threshold sequence and completed tables. Consequently,
\[
 E_j\cap\{q^{<j}=h\}=F_h\cap\{q^{<j}=h\}.
\]

The history bound~\eqref{eq:exp-history} still holds after conditioning on $F_h$. Indeed, conditional on $r,s^*$, the event $F_h$ depends only on the fresh round-$j$ coins, while $q^{<j}$ depends only on earlier coins. Conditional independence and~\eqref{eq:exp-history} therefore give, whenever the conditioning event has positive probability,
\[
 \Pr\!\left[q^{<j}=h\mid R=r,s^*,F_h\right]
 =\Pr\!\left[q^{<j}=h\mid R=r,s^*\right]
 \le D^{j-1}\mu^{\otimes(j-1)}(h).
\]
If $\Pr[F_h]=0$, then $\Pr[E_j,q^{<j}=h]=0$ as well. Otherwise, averaging the preceding bound over $R,s^*$ conditional on $F_h$ yields
\[
 \Pr\!\left[q^{<j}=h\mid F_h\right]
 \le D^{j-1}\mu^{\otimes(j-1)}(h).
\]
Factoring the joint probability with $F_h$ first now gives
\begin{align}
 \Pr\!\left[E_j,\ q^{<j}=h\right]
 &=\Pr\!\left[F_h,\ q^{<j}=h\right]\notag\\
 &=\Pr[F_h]\cdot\Pr\!\left[q^{<j}=h\mid F_h\right]
 &&\text{product rule}\notag\\
 &\le D^{j-1}\mu^{\otimes(j-1)}(h)\cdot\Pr[F_h]
 &&\text{preceding bound}\notag\\
 &\le N2^{N-1}D^{j-1}\mu^{\otimes(j-1)}(h)\sqrt{\frac{2N^n}{M}}
 &&\text{\Cref{lem:exp-coupling}}.
 \label{eq:exp-history-failure}
\end{align}
Finally, summing~\eqref{eq:exp-history-failure} over all histories and using $\sum_h\mu^{\otimes(j-1)}(h)=1$, we obtain
\begin{equation}\label{eq:exp-round-failure}
 \Pr[E_j]
 =\sum_h\Pr\!\left[E_j,\ q^{<j}=h\right]
 \le N2^{N-1}D^{j-1}\sqrt{\frac{2N^n}{M}}.
\end{equation}
A union bound over the $k$ rounds now yields
\begin{equation}\label{eq:exp-failure}
 \Pr\!\left[\bigcup_{j=1}^k E_j\right]
 \le kN2^{N-1}D^k\sqrt{\frac{2N^n}{M}}.
\end{equation}

\subsubsection{Choosing the universe and fixing the randomness}\label{subsec:exp-parameters}

We now choose $M$ large enough, to obtain small failure probability in Eq. \ref{eq:exp-failure}. The explicit choice
\begin{equation}\label{eq:exp-R}
 M=\left\lceil 2N^n\bigl(100kN2^ND^k\bigr)^2\right\rceil,
 \qquad D=(N+1)2^N,
\end{equation}
makes the right-hand side of~\eqref{eq:exp-failure} at most $1/200$. Indeed,
\[
 \sqrt{\frac{2N^n}{M}}\le\frac{1}{100kN2^ND^k},
 \qquad
 kN2^{N-1}D^k\sqrt{\frac{2N^n}{M}}
 \le\frac{kN2^{N-1}D^k}{100kN2^ND^k}=\frac1{200}.
\]

When no coupling fails, our queries force the mechanism to exhibit the bit-natural behavior required by the attack of \Cref{lem:bit_natural}. That attack forces an inaccurate answer within $k=O(n\log N)$ rounds. Hence an inaccurate answer occurs with probability at least $199/200$ in the joint experiment.

The universe size is dominated by the cost of adaptivity: the factor $D^{2k}$ in~\eqref{eq:exp-R} contributes $2k\log_2D$ to $\log_2|\X|$. For $N=20n$ and $k=\Theta(n\log n)$, the remaining factors contribute only $O(n\log N)$. Since $\log_2D=N+\log_2(N+1)=\Theta(N)$,
\[
 \log_2|\X|=2k\log_2D+O(n\log N)
          =O(kN)=O\bigl(n^2\log n\bigr),
\]
concluding the proof of \Cref{thm:plain-exponential-intro}.

\subsection{Mechanisms with limited randomness, proof of \Cref{thm:plain-limited-randomness}}\label{subsec:exp-randomness}

The proof extends to mechanisms with a bounded number of private random bits. We include the random seed among the candidates tracked by the analyst, and record the response bits for every seed in each profile.
Recall that \Cref{thm:plain-limited-randomness} states that if the mechanism has at most $b$ private random bits, then using domain size $2^{O(n^2\log n+nb)}$, there is an attacker forcing error $\geq 0.1$ in $O(n\log n+b)$ rounds, with probability at least $9/10$.

\begin{proof}[Proof of \Cref{thm:plain-limited-randomness}]
View the mechanism's randomness as a uniform seed $\sigma\in\{0,1\}^b$, drawn independently of the sample. Fixing this seed makes the mechanism deterministic. Set $N=20n$ as before. The analyst now maintains the sample--seed pairs $(s,\sigma)\in[N]^n\times\{0,1\}^b$ consistent with the transcript, simulating all seeds without knowing the actual one. Define a common point to be a row appearing in at least $8/10$ of the surviving pairs. Each pair contributes only $n$ sample entries, so the counting argument in~\eqref{eq:common_points} still gives $|X_m|\le5n/4$. Thus the separating-query argument underlying \Cref{lem:bit_natural} applies to these pairs: whenever adjacent response bits agree for every pair omitting the changed point, the analyst can select a query for which every accurate answer eliminates at least $1/10$ of the surviving pairs. There are initially $2^bN^n$ pairs, so an inaccurate answer is forced within
\[
 k=\left\lceil\log_{10/9}(2^bN^n)\right\rceil+1=O(n\log n+b)
\]
rounds if no coupling fails.

To force this agreement simultaneously for all seeds, replace~\eqref{eq:exp-profile} by
\[
 V_i^h(q)=\bigl(\bit_{s,\sigma}(q;h)\bigr)_{
 (s,\sigma)\in([N]\setminus\{i\})^n\times\{0,1\}^b}
 \in\{0,1\}^{2^b(N-1)^n},
\]
where $\bit_{s,\sigma}(q;h)$ is the rounded response to $q$ after replaying $h$ on the sample indexed by $s$, using seed $\sigma$. For fixed $R_{-i}=r_{-i}$ and $h$, this profile is still independent of $R_i$. The proof of \Cref{lem:exp-coupling} therefore applies with $d=2^b(N-1)^n$, giving a single-round failure probability at most $N2^{N-1}\sqrt{2^{b+1}N^n/M}$.

The selected-query density bound still holds with $D=(N+1)2^N$. For the history bound and comparison in \Cref{subsec:exp-adaptive}, also fix the actual mechanism seed, and use this same seed in the interactions with $\AAA$ and $\AAA_h$. In the interaction with $\AAA_h$, the locations remain uniform after fixing the index sample and seed. Conditional on the locations, index sample, and seed, $F_h$ depends only on fresh analyst coins, while $q^{<j}$ depends only on past coins. Averaging over both the index sample and seed gives, exactly as in~\eqref{eq:exp-failure},
\[
 \Pr\!\left[\bigcup_{j=1}^k E_j\right]
 \le kN2^{N-1}D^k\sqrt{\frac{2^{b+1}N^n}{M}}.
\]
Choose
\[
 M=\left\lceil2^{b+1}N^n\bigl(100kN2^ND^k\bigr)^2\right\rceil.
\]
The failure probability is then at most $1/200$. Since $\log D=\Theta(N)$, the universe has size
\[
 \log|\X|=O(b+n\log N+kN)=O(n^2\log n+nb).
\]
\end{proof}

\section*{AI disclosure}
We used ChatGPT (Astra 6) to obtain the proof in
\Cref{sec:plain-exponential}. We edited and verified the proof
and take full responsibility for its correctness.

\section*{Acknowledgments}

\paragraph{Edith Cohen:}  Partially supported by the Israel Science Foundation (grant 1156/23). 

\paragraph{Haim Kaplan:} Partially supported by the Israel Science Foundation (grant 1156/23), and the Blavatnik Research Foundation.

\paragraph{Yishay Mansour:} This project has received funding from the European Research Council (ERC) under the European Union’s Horizon 2020 research and innovation program (grant agreement No. 882396), by the Israel Science Foundation, the Yandex Initiative for Machine Learning at Tel Aviv University and a grant from the Tel Aviv University Center for AI and Data Science (TAD). 

\paragraph{Uri Stemmer:} Partially supported by the Israel Science Foundation (grant 1419/24), and the Blavatnik Research Foundation.

\bibliographystyle{alpha}
\bibliography{biblio}

@inproceedings{DworkFHPRR15,
  author       = {Cynthia Dwork and
                  Vitaly Feldman and
                  Moritz Hardt and
                  Toniann Pitassi and
                  Omer Reingold and
                  Aaron Leon Roth},
  title        = {Preserving Statistical Validity in Adaptive Data Analysis},
  booktitle    = {{STOC}},
  pages        = {117--126},
  publisher    = {{ACM}},
  year         = {2015}
}

@inproceedings{BassilyNSSSU16,
  author       = {Raef Bassily and
                  Kobbi Nissim and
                  Adam D. Smith and
                  Thomas Steinke and
                  Uri Stemmer and
                  Jonathan R. Ullman},
  title        = {Algorithmic stability for adaptive data analysis},
  booktitle    = {{STOC}},
  pages        = {1046--1059},
  publisher    = {{ACM}},
  year         = {2016}
}

@inproceedings{HardtU14,
  author       = {Moritz Hardt and
                  Jonathan R. Ullman},
  title        = {Preventing False Discovery in Interactive Data Analysis Is Hard},
  booktitle    = {{FOCS}},
  pages        = {454--463},
  publisher    = {{IEEE} Computer Society},
  year         = {2014}
}

@inproceedings{SteinkeU15,
  author       = {Thomas Steinke and
                  Jonathan R. Ullman},
  title        = {Interactive Fingerprinting Codes and the Hardness of Preventing False
                  Discovery},
  booktitle    = {{COLT}},
  series       = {{JMLR} Workshop and Conference Proceedings},
  volume       = {40},
  pages        = {1588--1628},
  publisher    = {JMLR.org},
  year         = {2015}
}

@inproceedings{NissimSSSU18,
  author       = {Kobbi Nissim and 
                  Adam D. Smith and
                  Thomas Steinke and
                  Uri Stemmer and
                  Jonathan R. Ullman},
  title        = {The Limits of Post-Selection Generalization},
  booktitle    = {NeurIPS},
  pages        = {6402--6411},
  year         = {2018}
}

@phdthesis{stemmer2016,
  title={Individuals and privacy in the eye of data analysis},
  author={Stemmer, Uri},
  school={Ben-Gurion University of the Negev},
  year = {2016}
}

@inproceedings{FishRR18,
  author       = {Benjamin Fish and
                  Lev Reyzin and
                  Benjamin I. P. Rubinstein},
  title        = {Sublinear-Time Adaptive Data Analysis},
  booktitle    = {{ISAIM}},
  year         = {2018}
}

@inproceedings{Blanc23,
  author       = {Guy Blanc},
  title        = {Subsampling Suffices for Adaptive Data Analysis},
  booktitle    = {{STOC}},
  pages        = {999--1012},
  publisher    = {{ACM}},
  year         = {2023}
}

@inproceedings{NissimST23,
  author       = {Kobbi Nissim and
                  Uri Stemmer and
                  Eliad Tsfadia},
  title        = {Adaptive Data Analysis in a Balanced Adversarial Model},
  booktitle    = {NeurIPS},
  year         = {2023}
}

@inproceedings{DinurSWZ23,
  author       = {Itai Dinur and
                  Uri Stemmer and
                  David P. Woodruff and
                  Samson Zhou},
  title        = {On Differential Privacy and Adaptive Data Analysis with Bounded Space},
  booktitle    = {{EUROCRYPT} {(3)}},
  series       = {Lecture Notes in Computer Science},
  volume       = {14006},
  pages        = {35--65},
  publisher    = {Springer},
  year         = {2023}
}

@inproceedings{KontorovichSS22,
  author       = {Aryeh Kontorovich and
                  Menachem Sadigurschi and
                  Uri Stemmer},
  title        = {Adaptive Data Analysis with Correlated Observations},
  booktitle    = {{ICML}},
  series       = {Proceedings of Machine Learning Research},
  volume       = {162},
  pages        = {11483--11498},
  publisher    = {{PMLR}},
  year         = {2022}
}

@article{RapoportCS25,
  author       = {Emma Rapoport and
                  Edith Cohen and
                  Uri Stemmer},
  title        = {Tight Bounds for Answering Adaptively Chosen Concentrated Queries},
  journal      = {CoRR},
  volume       = {abs/2507.13700},
  year         = {2025}
}

@inproceedings{JungLN0SS20,
  author       = {Christopher Jung and
                  Katrina Ligett and
                  Seth Neel and
                  Aaron Roth and
                  Saeed Sharifi{-}Malvajerdi and
                  Moshe Shenfeld},
  title        = {A New Analysis of Differential Privacy's Generalization Guarantees},
  booktitle    = {{ITCS}},
  series       = {LIPIcs},
  volume       = {151},
  pages        = {31:1--31:17},
  publisher    = {Schloss Dagstuhl - Leibniz-Zentrum f{\"{u}}r Informatik},
  year         = {2020}
}

@inproceedings{LigettS19,
  author       = {Katrina Ligett and
                  Moshe Shenfeld},
  title        = {A Necessary and Sufficient Stability Notion for Adaptive Generalization},
  booktitle    = {NeurIPS},
  pages        = {11481--11490},
  year         = {2019}
}

@inproceedings{ShenfeldL23,
  author       = {Moshe Shenfeld and
                  Katrina Ligett},
  title        = {Generalization in the Face of Adaptivity: {A} Bayesian Perspective},
  booktitle    = {NeurIPS},
  year         = {2023}
}

@inproceedings{RussoZ16,
  author       = {Daniel Russo and
                  James Zou},
  title        = {Controlling Bias in Adaptive Data Analysis Using Information Theory},
  booktitle    = {{AISTATS}},
  series       = {{JMLR} Workshop and Conference Proceedings},
  volume       = {51},
  pages        = {1232--1240},
  publisher    = {JMLR.org},
  year         = {2016}
}

@inproceedings{SteinkeZ20,
  author       = {Thomas Steinke and
                  Lydia Zakynthinou},
  title        = {Reasoning About Generalization via Conditional Mutual Information},
  booktitle    = {{COLT}},
  series       = {Proceedings of Machine Learning Research},
  volume       = {125},
  pages        = {3437--3452},
  publisher    = {{PMLR}},
  year         = {2020}
}

@inproceedings{RogersRST16,
  author       = {Ryan M. Rogers and
                  Aaron Roth and
                  Adam D. Smith and
                  Om Thakkar},
  title        = {Max-Information, Differential Privacy, and Post-selection Hypothesis
                  Testing},
  booktitle    = {{FOCS}},
  pages        = {487--494},
  publisher    = {{IEEE} Computer Society},
  year         = {2016}
}

@inproceedings{DworkFHPRR15b,
  author       = {Cynthia Dwork and
                  Vitaly Feldman and
                  Moritz Hardt and
                  Toniann Pitassi and
                  Omer Reingold and
                  Aaron Roth},
  title        = {Generalization in Adaptive Data Analysis and Holdout Reuse},
  booktitle    = {{NeurIPS}},
  pages        = {2350--2358},
  year         = {2015}
}

@inproceedings{FeldmanS17,
  author       = {Vitaly Feldman and
                  Thomas Steinke},
  title        = {Generalization for Adaptively-chosen Estimators via Stable Median},
  booktitle    = {{COLT}},
  series       = {Proceedings of Machine Learning Research},
  volume       = {65},
  pages        = {728--757},
  publisher    = {{PMLR}},
  year         = {2017}
}

@inproceedings{FeldmanS18,
  author       = {Vitaly Feldman and
                  Thomas Steinke},
  title        = {Calibrating Noise to Variance in Adaptive Data Analysis},
  booktitle    = {{COLT}},
  series       = {Proceedings of Machine Learning Research},
  volume       = {75},
  pages        = {535--544},
  publisher    = {{PMLR}},
  year         = {2018}
}

@inproceedings{LyuT25,
  author       = {Xin Lyu and
                  Kunal Talwar},
  title        = {Fingerprinting Codes Meet Geometry: Improved Lower Bounds for Private
                  Query Release and Adaptive Data Analysis},
  booktitle    = {{STOC}},
  pages        = {2374--2385},
  publisher    = {{ACM}},
  year         = {2025}
}

@inproceedings{BeimelKMNSS22,
  author       = {Amos Beimel and
                  Haim Kaplan and
                  Yishay Mansour and
                  Kobbi Nissim and
                  Thatchaphol Saranurak and
                  Uri Stemmer},
  title        = {Dynamic algorithms against an adaptive adversary: generic constructions
                  and lower bounds},
  booktitle    = {{STOC}},
  pages        = {1671--1684},
  publisher    = {{ACM}},
  year         = {2022}
}

@inproceedings{bateni2023optimal,
  title={Optimal fully dynamic k-center clustering for adaptive and oblivious adversaries},
  author={Bateni, MohammadHossein and Esfandiari, Hossein and Fichtenberger, Hendrik and Henzinger, Monika and Jayaram, Rajesh and Mirrokni, Vahab and Wiese, Andreas},
  booktitle={Proceedings of the 2023 Annual ACM-SIAM Symposium on Discrete Algorithms (SODA)},
  pages={2677--2727},
  year={2023},
  organization={SIAM}
}

@inproceedings{BernsteinBFKS26,
  author       = {Aaron Bernstein and
                  Sayan Bhattacharya and
                  Nick Fischer and
                  Peter Kiss and
                  Thatchaphol Saranurak},
  title        = {Separations between Oblivious and Adaptive Adversaries for Natural
                  Dynamic Graph Problems},
  booktitle    = {{SODA}},
  pages        = {5669--5687},
  publisher    = {{SIAM}},
  year         = {2026}
}

@inproceedings{alon1996space,
  title={The space complexity of approximating the frequency moments},
  author={Alon, Noga and Matias, Yossi and Szegedy, Mario},
  booktitle={Proceedings of the twenty-eighth annual ACM symposium on Theory of computing},
  pages={20--29},
  year={1996}
}

@article{flajolet1985probabilistic,
  title={Probabilistic counting algorithms for data base applications},
  author={Flajolet, Philippe and Martin, G Nigel},
  journal={Journal of computer and system sciences},
  volume={31},
  number={2},
  pages={182--209},
  year={1985},
  publisher={Elsevier}
}

@article{kalyanasundaram1992probabilistic,
  title={The probabilistic communication complexity of set intersection},
  author={Kalyanasundaram, Bala and Schnitger, Georg},
  journal={SIAM Journal on Discrete Mathematics},
  volume={5},
  number={4},
  pages={545--557},
  year={1992},
  publisher={SIAM}
}

@inproceedings{kaplan2021separating,
  title={Separating adaptive streaming from oblivious streaming using the bounded storage model},
  author={Kaplan, Haim and Mansour, Yishay and Nissim, Kobbi and Stemmer, Uri},
  booktitle={Annual International Cryptology Conference},
  pages={94--121},
  year={2021},
  organization={Springer}
}

@inproceedings{yao1979some,
  title={Some complexity questions related to distributive computing (preliminary report)},
  author={Yao, Andrew Chi-Chih},
  booktitle={Proceedings of the eleventh annual ACM symposium on Theory of computing},
  pages={209--213},
  year={1979}
}

@inproceedings{lehmann1981advantages,
  title={On the advantages of free choice: A symmetric and fully distributed solution to the dining philosophers problem},
  author={Lehmann, Daniel and Rabin, Michael O},
  booktitle={Proceedings of the 8th ACM SIGPLAN-SIGACT symposium on Principles of programming languages},
  pages={133--138},
  year={1981}
}

@inproceedings{ben1983another,
  title={Another advantage of free choice (extended abstract) completely asynchronous agreement protocols},
  author={Ben-Or, Michael},
  booktitle={Proceedings of the second annual ACM symposium on Principles of distributed computing},
  pages={27--30},
  year={1983}
}

@article{fischer1985impossibility,
  title={Impossibility of distributed consensus with one faulty process},
  author={Fischer, Michael J and Lynch, Nancy A and Paterson, Michael S},
  journal={Journal of the ACM (JACM)},
  volume={32},
  number={2},
  pages={374--382},
  year={1985},
  publisher={ACM New York, NY, USA}
}

@inproceedings{ChakrabartiS24,
  author       = {Amit Chakrabarti and
                  Manuel Stoeckl},
  title        = {Finding Missing Items Requires Strong Forms of Randomness},
  booktitle    = {{CCC}},
  series       = {LIPIcs},
  volume       = {300},
  pages        = {28:1--28:20},
  publisher    = {Schloss Dagstuhl - Leibniz-Zentrum f{\"{u}}r Informatik},
  year         = {2024}
}

@inproceedings{chakrabarti2016strong,
  title={Strong fooling sets for multi-player communication with applications to deterministic estimation of stream statistics},
  author={Chakrabarti, Amit and Kale, Sagar},
  booktitle={2016 IEEE 57th Annual Symposium on Foundations of Computer Science (FOCS)},
  pages={41--50},
  year={2016},
  organization={IEEE}
}

@article{kamath2021simple,
  title={A simple proof of a new set disjointness with applications to data streams},
  author={Kamath, Akshay and Price, Eric and Woodruff, David P},
  journal={arXiv preprint arXiv:2105.11338},
  year={2021}
}

@inproceedings{harvey2008sketching,
  title={Sketching and streaming entropy via approximation theory},
  author={Harvey, Nicholas JA and Nelson, Jelani and Onak, Krzysztof},
  booktitle={2008 49th Annual IEEE Symposium on Foundations of Computer Science},
  pages={489--498},
  year={2008},
  organization={IEEE}
}

@article{littlestone1988learning,
  title={Learning quickly when irrelevant attributes abound: A new linear-threshold algorithm},
  author={Littlestone, Nick},
  journal={Machine learning},
  volume={2},
  number={4},
  pages={285--318},
  year={1988},
  publisher={Springer}
}

@inproceedings{stoeckl2023streaming,
  title={Streaming algorithms for the missing item finding problem},
  author={Stoeckl, Manuel},
  booktitle={Proceedings of the 2023 Annual ACM-SIAM Symposium on Discrete Algorithms (SODA)},
  pages={793--818},
  year={2023},
  organization={SIAM}
}

@inproceedings{bleichenbacher1998chosen,
  title={Chosen ciphertext attacks against protocols based on the RSA encryption standard PKCS\# 1},
  author={Bleichenbacher, Daniel},
  booktitle={Annual international cryptology conference},
  pages={1--12},
  year={1998},
  organization={Springer}
}

@inproceedings{manger2001chosen,
  title={A chosen ciphertext attack on RSA optimal asymmetric encryption padding (OAEP) as standardized in PKCS\# 1 v2. 0},
  author={Manger, James},
  booktitle={Annual international cryptology conference},
  pages={230--238},
  year={2001},
  organization={Springer}
}

@article{angluin1988queries,
  title={Queries and concept learning},
  author={Angluin, Dana},
  journal={Machine learning},
  volume={2},
  number={4},
  pages={319--342},
  year={1988},
  publisher={Springer}
}

@article{rivest1980coping,
  title={Coping with errors in binary search procedures},
  author={Rivest, Ronald L. and Meyer, Albert R. and Kleitman, Daniel J. and Winklmann, Karl and Spencer, Joel},
  journal={Journal of Computer and System Sciences},
  volume={20},
  number={3},
  pages={396--404},
  year={1980},
  publisher={Elsevier}
}

@article{pelc2002searching,
  title={Searching games with errors—fifty years of coping with liars},
  author={Pelc, Andrzej},
  journal={Theoretical Computer Science},
  volume={270},
  number={1-2},
  pages={71--109},
  year={2002},
  publisher={Elsevier}
}

@book{MRT18,
  author = {Mohri, Mehryar and Rostamizadeh, Afshin and Talwalkar, Ameet},
  title = {Foundations of Machine Learning},
  edition = {2},
  publisher = {MIT Press},
  year = {2018}
}

@book{levinPeres2026markov,
  title={Markov chains and mixing times},
  author={Levin, David A and Peres, Yuval},
  year={2026},
  publisher={American Mathematical Society}
}

\appendix
\section{Positive result for computationally bounded analysts}\label{sec:bounded_adversary}

Let us first show that the mechanism can extract randomness from the sample. We will then use a pseudorandom generator to stretch the randomness and obtain the quadratic bound.

\begin{lemma}\label{lem:extract_rand_permutation}
Let $\delta \in (0, 1), n \ge \Omega(\log^2\tfrac{1}{\delta})$ and domain $\mathcal{X}$ of size $|\X|=N\in \poly(n)$.
There exists a deterministic algorithm $\mathcal{M}_{extract}$ such that given a sample $S \sim \mathcal{P}^n$, with probability at least $1-\delta$, in $\poly(n)$ time, $\mathcal{M}_{extract}$ either:
\begin{enumerate}
    \item Extracts $m = \lfloor \sqrt{n} \rfloor$ uniformly random bits.
    \item Answers an unbounded number of adaptively chosen statistical queries $q_t \colon \mathcal{X} \to [0, 1]$ with $0.1$-accuracy.
\end{enumerate}
\end{lemma}
We remark that it seems possible that one can extract $O(n)$ random bits, but $\sqrt{n}$ bits suffice for our application.
\begin{proof}

The algorithm $\mathcal{M}_{extract}$ does the following:
Given $S = (x_1, \dots, x_n)$, compute the empirical multiplicity $N_x$ for all unique elements in $S$. Let $N_{\max} = \max_x N_x$, and let $\hat{y}$ be the corresponding most frequent element (breaking ties arbitrarily).
\begin{itemize}
    \item \textbf{Query Mode ($\boldsymbol{N_{\max} \geq 0.95 n}$):}
    For any adaptively chosen query $q_t$, output $a_t = q_t(\hat{y})$.
    \item \textbf{Extraction Mode ($\boldsymbol{N_{\max} < 0.95 n}$):}
    Let $m = \lfloor \sqrt{n} \rfloor$. Partition the first $m^2 \le n$ elements of $S$ into $m$ disjoint blocks $B_1, \dots, B_m$, each of size $m$. 
    
    For each block $i \in \{1, \dots, m\}$:
    \begin{itemize}
        \item Let $k = \lfloor m/2 \rfloor$. Let $L_i$ be the first $k$ elements of $B_i$, and $R_i$ be the next $k$ elements.
        \item Lexicographically compare the sequences $L_i$ and $R_i$. 
        \item If $L_i < R_i$, output bit $b_i = 1$.
        \item If $L_i > R_i$, output bit $b_i = 0$.
        \item If $L_i = R_i$, halt the algorithm and output \textbf{FAIL}.
    \end{itemize}
\end{itemize}

\paragraph{Probabilistic Analysis:}
We partition $\mathcal{X}$ into sets by a case analysis on the distribution $\mathcal{P}$. 
\begin{itemize}[leftmargin=37px]
    \item[Case 1:] Suppose there exists $y^*\in \mathcal{X}$ such that $\Pr_{Y\sim\mathcal{P}}[Y=y^*] \geq 0.9$. In this case, we consider the singleton $\{y^*\}$ and the set $A_1=\mathcal{X}\setminus \{y^*\}$.

    \item[Case 2:] For all $y\in\mathcal{X}$ we have $\Pr_{Y\sim\mathcal{P}}[Y=y]<0.9$ and there exists $y_1\in\mathcal{X}$ with $\Pr_{Y\sim\mathcal{P}}[Y=y_1]\geq 0.1$. In this case, we set $A_1=\{y_1\}$ and $A_2=\mathcal{X}\setminus A_1$.

    \item[Case 3:] Otherwise, sequence the domain $\mathcal{X}$ as $y_1,y_2,\ldots$ and let $S_m=\{y_i\mid i\leq m\}$. Consider the smallest integer $m$ such that $\Pr_{Y\sim\mathcal{P}}[Y\in S_m]\geq 0.1$. By minimality of $m$ and since all elements have probability mass at most $0.1$, we obtain $\Pr_{Y\sim\mathcal{P}}[Y\in S_m]\leq 0.2$. We set $A_1=S_m$ and $A_2=\mathcal{X}\setminus A_1$.
\end{itemize}

\noindent Note that in all cases we have $\Pr_{Y\sim\mathcal{P}}[Y\in A_1],\Pr_{Y\sim\mathcal{P}}[Y\in A_2]\leq 0.9$. 
Let $N(A_i)$ be the number of samples from $S$ belonging to $A_i$ (when $A_i$ is defined). By Hoeffding's Inequality:
\[
\Pr[N(A_i) \ge 0.95 n] \le \exp\left(-2n(0.05)^2\right) = \exp(-\Omega( n)).
\]
Since $n\geq\Omega(\log\tfrac{1}{\delta})$ and by a union bound, we obtain that with probability at least $1-\tfrac{\delta}{2}$, $N(A_1),N(A_2)<0.95n$. Denote this event by $\mathcal{E}$ and assume it holds.

\paragraph{Query Mode (${\boldsymbol{N_{\max} \ge 0.95 n}}$).}
If the algorithm routes to this mode, the most frequent element satisfies $N_{\hat{y}} \ge 0.95 n$. Under event $\mathcal{E}$, no non-canonical partition $A_i$ receives $0.95n$ samples, and thus $\hat{y}=y^*$. Therefore, $\Pr_{Y\sim\mathcal{P}}[Y=\hat{y}] \geq 0.90$.

For any adaptively chosen query $q_t \colon \mathcal{X} \to [0,1]$, 
we evaluate the error:
\[
\left| a_t - \mathbb{E}[q_t(Y)] \right| = \left| q_t(\hat{y}) - \Big( p_{\hat{y}} q_t(\hat{y}) + (1-p_{\hat{y}}) \mathbb{E}[q_t(Y) \mid Y \neq \hat{y}] \Big) \right| \le 1 - p_{\hat{y}} < 1 - 0.90 = 0.1.
\]
Because $\hat{y}$ is fixed statically and its state completely ignores the chosen queries, these strict $0.1$-error bounds hold uniformly and simultaneously over an unbounded number of adaptive evaluations.

\paragraph{Extraction Mode (${\boldsymbol{N_{\max} < 0.95n}}$).}
Condition on $N_{\max} < 0.95n$. Observe that every permutation of the dataset $S$ is equally likely.

\textit{1. Negligible Failure Probability:}
We bound the probability of outputting \textbf{FAIL}. Because a uniformly random permutation is exchangeable, the marginal distribution of elements populating any block $B_i$ is identical. We evaluate the probability $\Pr[L_1 = R_1]$ by imagining we sequentially draw its $2k$ elements in paired positions $(1, k+1), (2, k+2), \dots, (k, 2k)$.

The probability that the $j$-th pair matches, conditioned on the history of previous draws, is bounded by the maximum possible fraction of remaining identical elements. When drawing the second element of the $j$-th pair, we have already removed $2j - 1$ elements from the pool. In the worst case, none of the removed elements were of the majority class, leaving $N_{\max}$ identical elements in a pool of size $n - 2j + 1$. Because $2j \le m \le \sqrt{n}$ and $N_{\max} < 0.95n$, this probability is bounded:
\[
\Pr[\text{pair } j \text{ matches} \mid \text{history}] \le \frac{N_{\max}}{n - 2j + 1} \le \frac{0.95n}{n - \sqrt{n}} \leq 0.975,
\]
where the last step is by assuming $n \ge 1600$. 
The probability that all $k$ sequentially drawn pairs match is at most $(0.975)^k$. Applying a union bound over all $m$ disjoint blocks, the probability of any block failing is:
\[
\Pr[\text{\textbf{FAIL}}] \le m (0.975)^k \le \sqrt{n} (0.975)^{\lfloor \sqrt{n}/2 \rfloor}\leq \tfrac{\delta}{2},
\]
where the last step is since $n=\Omega(\log^2 \tfrac{1}{\delta})$.
By a union bound, with probability at least $1-\delta$, the algorithm does not output \textbf{FAIL} and $\mathcal{E}$ holds.

\textit{2. Perfect Uniformity:}
Consider the set of valid permutations of $S$ where no block ties (i.e., $L_i \ne R_i$ for all $i$).
Clearly, for all $i$, $\Pr[L_i<R_i]=\Pr[L_i>R_i]=\tfrac{1}{2}$, hence the $i$-th bit $b_i$ is uniform in $\{0,1\}$. Now, we claim that the bits $(b_1,\ldots,b_m)$ are independent.

Consider the realization of the blocks $\{(L_i,R_i)\mid i\leq m\}$. For every $i$, the permutation where the blocks $L_i,R_i$ are swapped, i.e., with $L'_i=R_i$ and $R'_i=L_i$, is equally likely.
Moreover, there are $2^m$ options for these swapping ($2$ options for each $i\in [m]$), and each one is equally likely. This holds regardless of the block's content, hence the bits $(b_1,\ldots,b_m)$ are uniform over $\{0,1\}^m$, concluding the proof.

\end{proof}

To obtain the positive result against computationally bounded analysts, we assume the existence of \emph{non-uniform} pseudorandom generators. Whenever we say (here and in the introduction) assuming PRGs exist, we mean this definition.
\begin{definition}
    Let $\ell \colon \mathbb{N} \to \mathbb{N}$ be a polynomial function such that $\ell(n) > n$ for all $n \in \mathbb{N}$ (the expansion factor). A deterministic, polynomial-time algorithm $G \colon \{0,1\}^* \to \{0,1\}^*$ is a \emph{Pseudorandom Generator (PRG)} if it satisfies the following two conditions:
    \begin{itemize}
        \item Expansion: For every input string $s \in \{0,1\}^n$, the output $G(s)$ has length exactly $\ell(n)$.
        \item Computational Indistinguishability: For every Probabilistic Polynomial-Time (PPT) algorithm $D$ (called a distinguisher) that takes an advice $a_n\in \{0,1\}^{\poly(n)}$, there exists a negligible function $\mathsf{negl}(n)$ such that:
        $$ \left| \Pr_{s \sim U_n}[D(G(s),a_n) = 1] - \Pr_{r \sim U_{\ell(n)}}[D(r,a_n) = 1] \right| \le \mathsf{negl}(n), $$
        where $U_n$ is the uniform distribution over $\{0,1\}^n$.
    \end{itemize}
\end{definition}
A standard assumption in cryptography is that for every polynomial $\ell$, there exists a PRG with expansion factor $\ell$.
Lastly, we recall that randomized mechanisms can answer a quadratic number of queries.
\begin{lemma}[\cite{BassilyNSSSU16}]\label{lem:rand_mechanism}
    For every constant $0<\eps_0<1$,
    there exists a randomized mechanism running in polynomial time such that given a dataset $S=(x_1,\ldots,x_n)\sim \mathcal{P}^n$, with high probability the algorithm answers $\tilde{\Omega}(n^2)$ adaptive statistical queries within accuracy $\eps_0$.
\end{lemma}

We are now ready to prove the positive result against computationally bounded analysts. 

\begin{observation}
    Assuming PRGs exist, there exists a deterministic mechanism $\MMM$ such that the following holds. For every sufficiently large $n\in\N$ and every computationally bounded analyst, when $\MMM$ is given a dataset $S\sim\mathcal{P}^n$ drawn iid from a distribution $\mathcal{P}$, with high probability, $\MMM$ answers $\tilde{\Omega}(n^2)$ adaptive statistical queries within accuracy $0.1$.
\end{observation}

\begin{proof}
    The mechanism does the following. First, employ the algorithm $\mathcal{M}_{extract}$ of \Cref{lem:extract_rand_permutation} on $(x_1,\ldots,x_{n/2})$. If $\mathcal{M}_{extract}$ enters the query mode, use its output. If $\mathcal{M}_{extract}$ enters the extraction mode and outputs a string $s$ of $\sqrt{n/2}$ bits, apply a PRG $G$ to expand $s$ to length $\ell=\poly(n)$. Apply the mechanism $\MMM_{rand}$ of \Cref{lem:rand_mechanism} with accuracy parameter $\eps_0=0.09$ on $(x_{n/2+1},\ldots,x_n)$ using $G(s)$.

    Let us assume the event in \Cref{lem:extract_rand_permutation} is successful. Therefore, if $\mathcal{M}_{extract}$ enters the query mode, we are done.
    Assume henceforth that $\mathcal{M}_{extract}$ enters the extraction mode. By the assumption that the event in \Cref{lem:extract_rand_permutation} is successful, $\mathcal{M}_{extract}$ outputs a string $s$ of $\sqrt{n/2}$ random bits. Clearly, $G(s)$ is independent of $(x_{n/2+1},\ldots,x_n)$. 
    Let an analyst $\AAA_{bounded}$ running in polynomial time. We claim that with high probability, the application of $\MMM_{rand}$ on $(x_{n/2+1},\ldots,x_n)$ using the expanded randomness $G(s)$ answers $\tilde{\Omega}(n^2)$ adaptive queries issued by $\AAA_{bounded}$ within $0.1$-accuracy. 

    Assume by contradiction that when using the pseudo-random string $G(s)$, the mechanism $\MMM_{rand}$ answers fewer than the $\tilde{\Omega}(n^2)$ bound of \Cref{lem:rand_mechanism} adaptive queries issued by $\AAA_{bounded}$, i.e., with non-negligible probability, one of the answers has error $>0.1$.
    We now use the distribution $\mathcal{P}$, the analyst $\AAA_{bounded}$ and the mechanism $\MMM_{rand}$ to construct a distinguisher $D$ that distinguishes between $G(s)$ and a string drawn uniformly from $\{0,1\}^\ell$.

    \paragraph{Samplable Distributions.}
    Let us make a simplifying assumption that $\mathcal{P}$ is samplable, i.e., one can sample from $\mathcal{P}$ in $\poly(n)$ time. We will later show how to remove this assumption. Denote the distinguisher that uses this assumption by $D'$.

    Given an input string $r \in \{0,1\}^{\ell}$, the distinguisher $D'$ operates as follows. First, $D'$ draws a fresh dataset $S' = (x_{n/2+1}, \dots, x_n) \sim \mathcal{P}^{n/2}$ and a large validation dataset $V \sim \mathcal{P}^{n'}$, where $n' =O(\log n)$, chosen to be sufficiently large to estimate query expectations to within a constant error of $\eps_{small}=\eps_0/100$ via standard concentration bounds. Next, $D'$ simulates the interaction between the computationally bounded analyst $\mathcal{A}_{bounded}$ and the mechanism $\mathcal{M}_{rand}$. In this simulation, $\mathcal{M}_{rand}$ uses $S'$ as its input dataset and the string $r$ as its internal randomness. 
    For every $t\in[k]$, denote by $q_t,a_t$ the queries of $\AAA_{bounded}$ and answers of $\MMM_{rand}$, respectively.
    Denote by $\mu_t$ the empirical expectation of the query over the validation set, i.e., $\mu_t=\tfrac{1}{n'}\sum_{x \in V} q_t(x)$.
    If there exists round $t \in [k]$ where $|a_t - {\mu}_t| > \eps_0+\eps_{small}$, the distinguisher $D'$ outputs $1$ (indicating $r$ is likely pseudorandom), and otherwise $D'$ output $0$ (indicating $r$ is likely truly random).
    
    Observe that $V$ is independent of $q_1,\ldots,q_k$. By setting $n'$ sufficiently large, we have by standard concentration bounds and a union bound that, with high probability, 
    \[\forall t\in[k], \qquad |\mu_t-\E_{x\sim\mathcal{P}}[q_t(x)]|\leq \eps_{small}.\] 
Hence, if $r$ is pseudo-random, then by assumption, there exists $t\in[k]$ such that $|a_t-\E_{x\sim\mathcal{P}}[q_t(x)]|>0.1\geq \eps_0+2\eps_{small}$. Therefore, 
    $|a_t-\mu_t|>\eps_0+\eps_{small}$ and $D'$ outputs $1$.
    On the other hand, if $r$ is truly random, we have by \Cref{lem:rand_mechanism} that for all $t\in[k]$,
    $|a_t-\E_{x\sim\mathcal{P}}[q_t(x)]|\leq \eps_0$, hence $|a_t - \tilde{\mu}_t| \leq \eps_0+\eps_{small}$ and $D'$ outputs $0$.
    Therefore, $D'$ distinguish between truly random $r\in\{0,1\}^\ell$ and pseudo-random $r=G(s)$ where $s\in\{0,1\}^{\sqrt{n/2}}$ with non-negligible probability, contradiction. Therefore, the result hold when the distribution $\mathcal{P}$ is samplable.

    \paragraph{General Distributions, Using the Advice.}
    Recall the setting: we fixed the mechanism $\MMM_{rand}$, analyst $\AAA_{bounded}$ and distribution $\mathcal{P}$, and assume that when given true random string $r\in\{0,1\}^\ell$, the mechanism is $\eps_0$-accurate, and when given pseudo-random string, one of the answers has error $>0.1$.
    Construct an advice $a_n=(S^*,V^*)$, with $S^*\in \X^{n/2}$ and $V^*\in \X^{n'}$ maximizing 
    \[\left| \Pr_{s \sim U_{\sqrt{n/2}}}[D(G(s),(S,V)) = 1] - \Pr_{r \sim U_{\ell}}[D(r,(S,V)) = 1] \right|.\]
    A distinguisher $D$ obtains this advice, and implements $D'$ assuming it has drawn the sample $S^*$ and the validation set $V^*$.    
Observe that if $D$ were to actually draw $S\sim \mathcal{P}^{n/2}$ and $V\sim \mathcal{P}^{n'}$, then it would have precisely implemented $D'$. By the same arguments as before,  it would then distinguishes with non-negligible probability between truly random $r\in\{0,1\}^\ell$ and pseudo-random $r=G(s)$ where $s\in\{0,1\}^{\sqrt{n/2}}$.
Clearly, the maximum is at least as large as the average, and thus $D$, with the advice $a_n$, distinguishes with non-negligible probability between truly random $\ell$-bit string and a pseudo-random one, contradiction.
Therefore, $\MMM$ is $0.1$-accurate, concluding the proof.
    \end{proof}

\end{document}